\documentclass[lettersize,journal]{IEEEtran}
\usepackage{amsmath,amsfonts}
\usepackage{algorithmic}
\usepackage{algorithm}
\usepackage{array}
\usepackage[caption=false,font=tiny,labelfont=rm,textfont=rm]{subfig}
\usepackage{textcomp}
\usepackage{stfloats}
\usepackage{url}
\usepackage{verbatim}
\usepackage{graphicx}
\usepackage{cite}
\usepackage{multirow}
\usepackage{array}
\usepackage{enumitem} 
\usepackage{amsthm}
\usepackage{tabularx}  
\usepackage{booktabs}

\usepackage{amsthm}
\usepackage{subfig}
\usepackage{caption}
\usepackage{times}
\usepackage{mathrsfs}
\usepackage{hyperref}

\usepackage{booktabs}
\usepackage{amsthm}

\newtheoremstyle{mylemma}
{0.5\topsep}
{0.5\topsep}
{\normalfont}
{}
{\itshape}             % 这里改成同时作用于名称和编号
{: }
{ }
{}

\theoremstyle{mylemma}
\newtheorem{lemma}{Lemma}

\usepackage{bm}
\newtheoremstyle{observation}%     样式名称
{3pt}%                          上间距
{3pt}%                          下间距
{\itshape}%                     定理体样式（斜体）
{0.5em}%                             缩进
{\bfseries\itshape}%            标题样式（加粗+斜体）
{: }%                           标题后的标点（冒号+空格）
{ }%                            标题后的空格
{}%                             标题后代码

\theoremstyle{observation}
\newtheorem{observation}{\textbf{\textit{}}}

\usepackage{caption}
\usepackage{placeins}

\usepackage{float} % 添加这个包

\usepackage{orcidlink}

\usepackage{amsthm}

\newtheoremstyle{remarkstyle}
{3pt}% Space above
{3pt}% Space below
{\normalfont}% Body font (正体)
{}% Indent amount
{\bfseries\itshape}% Theorem head font (加粗+斜体)
{:}% Punctuation after theorem head (冒号)
{.5em}% Space after theorem head
{\bfseries\itshape\thmname{#1}\thmnumber{ #2}}% Theorem head spec (加粗+斜体)

\theoremstyle{remarkstyle}
\newtheorem{remark}{Remark}

\makeatletter
\renewenvironment{proof}[1][\proofname]{\par
	\pushQED{\qed}%
	\normalfont \topsep6\p@\@plus6\p@\relax
	\trivlist
	\item[\hskip\labelsep
	\itshape
	#1:]\ignorespaces  % 将点改为冒号
}{%
	\popQED\endtrivlist\@endpefalse
}
\makeatother

\begin{document}
\title{Specific Multi-emitter Identification: Theoretical Limits and Low-complexity Design}
%\author{Yuhao Chen~\IEEEmembership{}
%        % <-this % stops a space
%% <-this % stops a space  This paper was produced by the IEEE Publication Technology Group. They are in Piscataway, NJ.
%}%\thanks{}

\author{Yuhao Chen\orcidlink{0009-0000-5050-1688},
	Boxiang He\orcidlink{0000-0002-9235-1144},
	Junshan Luo\orcidlink{0000-0003-4151-5126},
	Shilian Wang\orcidlink{0000-0003-4132-8750},
	Lei Yao \orcidlink{0009-0005-3172-1462},
	and Jing Lei\orcidlink{0000-0002-5838-5826}\thanks{A preliminary version of this work has been submitted to the IEEE International Conference on Acoustics, Speech and Signal Processing (ICASSP 2026), Paper ID 16896, and is currently under review \cite{Chen2025SpecificMultiEmitterIdentification}. This journal submission significantly extends the preliminary version with additional theoretical results, experiments, and discussions.}
	\thanks{Y. Chen, B. He, J. Luo, S. Wang, L. Yao, and J. Lei are with the College of Electronic Science and Technology, National University of Defense Technology, Changsha 410003, P. R. China. (email: cyh20220720@163.com; boxianghe1@bjtu.edu.cn; luojunshan10@nudt.edu.cn; wangsl@nudt.edu.cn; yaolei11103@163.com; leijing@nudt.edu.cn).}
	}
\markboth{}%Journal of \LaTeX\ Class Files,~Vol.~14, No.~8, August~2021
{Shell \MakeLowercase{\textit{et al.}}: A Sample Article Using IEEEtran.cls for IEEE Journals}
\maketitle

\begin{abstract}
Specific emitter identification (SEI) distinguishes emitters by utilizing hardware-induced signal imperfections. However, conventional SEI techniques are primarily designed for single-emitter scenarios. This poses a fundamental limitation in distributed wireless networks, where simultaneous transmissions from multiple emitters result in overlapping signals that conventional single-emitter identification methods cannot effectively handle. To overcome this limitation, we present a specific multi-emitter identification (SMEI) framework via multi-label learning, treating identification as a problem of directly decoding emitter states from overlapping signals. Theoretically, we establish performance bounds using Fano's inequality. Methodologically, the multi-label formulation reduces output dimensionality from exponential to linear scale, thereby substantially decreasing computational complexity. Additionally, we propose an improved SMEI (I-SMEI), which incorporates multi-head attention to effectively capture features in correlated signal combinations. Experimental results demonstrate that SMEI achieves high identification accuracy with a linear computational complexity. Furthermore, the proposed I-SMEI scheme significantly improves identification accuracy across various overlapping scenarios compared to the proposed SMEI and other advanced methods.
\end{abstract}


\section{Introduction}
\label{sec1}
\IEEEPARstart{S}{pecific} emitter identification (SEI) exploits hardware-level imperfections to uniquely identify individual transmitters~\cite{Yi2024SpecificEmitterIdentification}. Due to variations in electrical components between different devices, the signals emitted by different radios can vary measurably. By analyzing received signals and extracting these features, it is possible to effectively distinguish between different emitters~\cite{Chen2019RFFingerprint}. Leveraging these advantages, SEI serves as a versatile physical-layer security mechanism applicable to various scenarios, including cognitive radio authentication~\cite{Liu2024CognitiveJammers}, intrusion prevention in high-speed mobile systems~\cite{Ma2025OrthogonalDelayDoppler}, intelligent attacker identification~\cite{He2025SecureSemanticTransmission}, and secure anti-jamming in low-altitude wireless networks~\cite{Guo2025DualEndFluidAntenna}.

The objective of SEI is to employ classifiers for identification after extracting the radio frequency (RF) fingerprint from these signals \cite{AlShawabka2020Exposing},  \cite{Alhoraibi2023PhysicalLayerAuthentication}, \cite{Zhao2022ComplexValuedLearning}. Recent SEI research integrates deep learning into identification frameworks \cite{Wu2023, Gao2025}, enabling direct feature extraction and classification from raw signals. To improve identification accuracy, a dual-channel feature-extraction network, for instance, combines variational mode decomposition with time-frequency and temporal features using convolutional neural networks (CNNs). It also incorporates an attention module and a dual-modal feature fusion strategy \cite{Su2023SpecificEmitterIdentification}. Moreover, approaches using dynamic neural networks, recurrent neural networks (RNNs), and CNNs for classification, and generative adversarial networks for feature extraction, achieve accuracies of 81.6\%, 94.6\%, and 97.7\%, respectively \cite{Roy2020RFAL}. Integrating classical characteristics such as bispectra, cyclic spectra, and power spectra into a CNN, with intrinsic features derived from raw signals via RNNs, can improve identification accuracy to 92.25\% \cite{Ying2022ChannelAttention}. To further enhance learning performance, current research investigates preprocessing methods that convert signals into visual representations, in addition to directly feeding raw signals into deep learning models. For example, drone signal characteristics are extracted using the short-time Fourier transform (STFT) and subsequently classified using a CNN \cite{Yu2024DroneRFa}. By extending this method, dual-stream networks achieve over 95\% classification accuracy at signal-to-noise ratios (SNRs) between 10 dB and 20 dB by combining physics-informed impulse decomposition with local STFT attention to effectively model hardware impairments, such as in-phase and quadrature (I/Q) imbalances \cite{Zhao2025PhysicsInformed}. Likewise, robust emitter identification is possible when time-frequency distribution maps generated by the STFT are fed into CNNs \cite{Wang2017RadarEmitterRecognition}. Additionally, the Hilbert-Huang transform is applied to received signals, yielding grayscale images that deep residual networks can process for identification and learning, yielding positive outcomes \cite{Pan2019SpecificEmitterIdentification}. 

Although deep learning and time-frequency analysis approaches significantly improve SEI performance, most SEI research focuses on single-emitter scenarios that fail to address the core requirements of modern wireless systems with distributed network architectures. In practical applications such as industrial, scientific, and medical bands, where simultaneous transmissions are essential, multiple devices inherently transmit concurrently \cite{Forenbacher2021,Polak2022}. These emitters cause received signals to be mixtures of various sources, presenting challenges that conventional single-emitter identification methods cannot adequately address. Frequently, they transmit in tightly spaced time intervals or occupy neighboring sub-bands, which causes received signals to be overlapping \cite{Zeng2019UAV5G,Baldini2019Impact}. For instance, the frequent occurrence of multiple signals occupying the same frequency band in drone-assisted radio monitoring systems leads to overlapping spectrum signals, making signal identification more difficult \cite{Hao2025}. The extension of SEI to specific multi-emitter identification (SMEI), which seeks to identify the complete set of active distributed emitters from overlapping signals, is motivated by this practical issue. Fig.~\ref{SMEIPROCESS1} summarizes the fundamental framework of SMEI. When the SMEI system receives signals with time-frequency overlap, it first extracts multi-emitter RF fingerprints, performs classification, and then uses a threshold-based decision procedure to identify the set of active distributed emitters.

One approach to SMEI is to structure the problem as multiclass classification over all non-empty subsets of $K$ devices, which produces an output space of $2^K-1$ classes \cite{Sankhe2019ORACLE}. However, this exponential growth in $K$ poses significant challenges for inference and training, including a sharp increase in model parameters due to the different potential combinations \cite{Chen2021RFMixtureDDL}. Furthermore, it is challenging to extract reliable multi-emitter signatures because receiver noise or channel effects can quickly disguise the delicate RF fingerprint characteristics that distinguish devices \cite{Jin2023RFFeatureFusion}. These challenges are exacerbated in low-SNR scenarios, where both feature-based and deep learning techniques often exhibit severe performance declines. In distributed wireless environments with signal overlap and noise, traditional single-emitter identification and multiclass frameworks are therefore insufficient. Specifically, effective identification is hampered by the combined problems of feature extraction constraints and an explosion of the output space. Inspired by these shortcomings and building upon our preliminary multi-label learning framework \cite{Chen2025SpecificMultiEmitterIdentification}, this research investigates novel approaches for SMEI that aim to enhance feature disentanglement and robustness.
\begin{figure*}[htbp]
	\centering
	\includegraphics[width=0.7\textwidth]{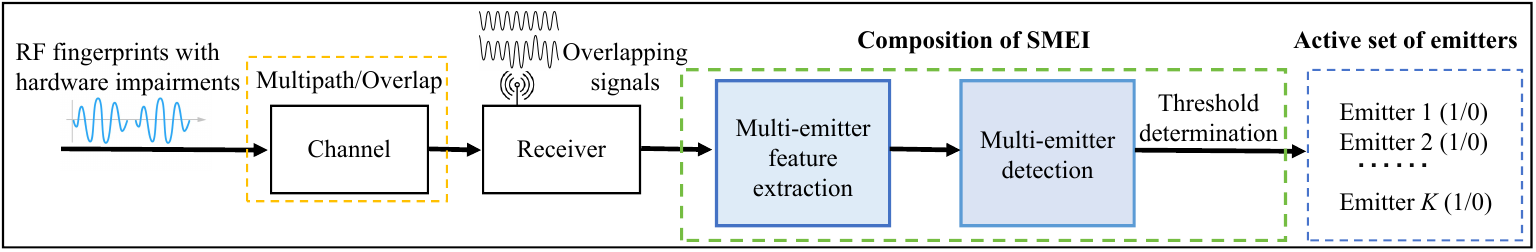}
	\caption{The RF fingerprint-derived multi-emitter identification system framework comprises a comprehensive processing pipeline. This pipeline includes multi-emitter feature extraction, detection, and threshold determination.}
	\label{SMEIPROCESS1}
\end{figure*}
\begin{figure*}[htbp]
	\centering
	\includegraphics[width=0.7\textwidth]{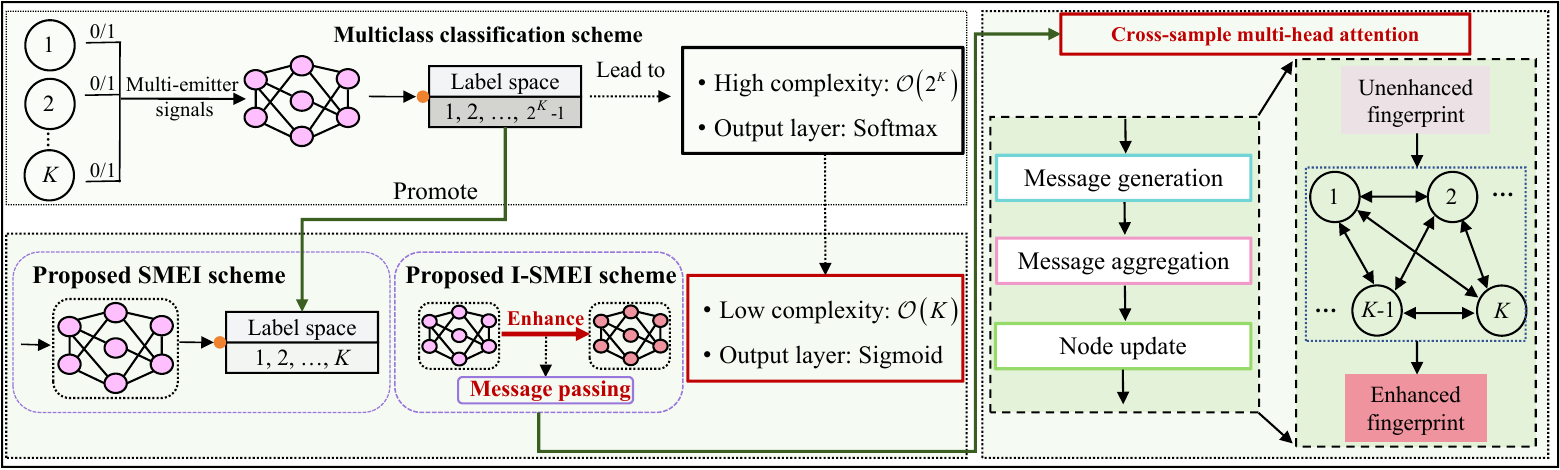}
	\caption{The proposed SMEI framework addresses the high complexity of multiclassification by reducing the label space dimensionality and enhancing fingerprint extraction through message passing.}
	\label{Technical_approach_3}
\end{figure*}

To overcome the aforementioned complexity issues in emitter identification for distributed networks, we introduce a low-complexity multi-emitter identification approach based on multi-label learning, as depicted in Fig.~\ref{Technical_approach_3}. The key contributions of this paper are outlined below:

\begin{itemize}
	\item We first analyze the theoretical performance limits of the SMEI problem from an information-theoretic perspective. By deriving upper bounds for both subset and Hamming accuracies using Fano's inequality, we establish a theoretical benchmark for evaluating model performance.
	
	\item The SMEI problem is addressed through a multi-label classification framework that treats each emitter as an independent label and enables direct prediction of the active device set. This design reduces the output dimension from the exponential scale of combinatorial classes $\left(2^K-1\right)$ to a linear scale $\left(K\right)$, significantly improving system scalability while mitigating model parameter explosion.

	\item To further enhance the model's ability to distinguish overlapping fingerprint features in overlapping signals, we propose an improved SMEI (I-SMEI), which incorporates a message-passing mechanism via multi-head attention. This mechanism enables information interaction at the feature level and allows the model to dynamically capture the interdependencies among different emitter features during the mixing process. This enhances feature disentanglement for overlapping signals.
	
	\item Experimental results show that the proposed SMEI approach achieves identification accuracy comparable to that of traditional multiclass classification methods while requiring significantly less complexity. It also outperforms STFT-based  identification methods. Additionally, I-SMEI outperforms the Query2Label (Q2L) enhancement approach in terms of robustness and identification accuracy across a wide range of channel scenarios.
\end{itemize}

This paper continues as follows: The system model is introduced in Section \ref{sec2}. Section \ref{sec3} derives the theoretical upper bounds for identification accuracy. The core design of the multi-label learning architecture is presented in Section \ref{sec4}, and the proposed I-SMEI scheme is detailed in Section \ref{sec5}. Subsequently, Section \ref{sec6} is devoted to the experimental setup and results. The paper concludes in Section \ref{sec7}.

\emph{Notation:} Throughout this paper, scalars, vectors, and matrices are denoted by lower-case italic letters $x$, bold lower-case italic letters $\bm{x}$, and bold capital italic letters $\bm{X}$, respectively. A random variable and its realization are respectively written as $\mathsf{x}$ and $x$. The operators $[\cdot]^{\mathsf{T}}$ and $[\cdot]^{\dagger}$ denote the transpose and conjugate transpose, respectively. The symbol $\odot$ represents the Hadamard (element-wise) product, while $\log(\cdot)$ denotes the logarithm. The notation $\mathcal{CN}\left({\mu},\varOmega^2\right)$ denotes the probability density function of a random variable following the complex Gaussian distribution with mean $\mu$ and variance $\varOmega^2$. $|\cdot|$ denotes the total number of distinct categories or types in a given set. $H(\cdot)$ and $I(\cdot;\cdot)$ represent the entropy function and mutual information function, respectively. $\mathcal{I}(\cdot)$ is the indicator function that returns one if the condition is true and zero otherwise. $g^{-1}(\cdot)$ denotes the inverse function of $g(\cdot)$. $\min(\cdot,\cdot)$ represents the minimum value function. Let $j=\sqrt{-1}$. $\mathcal{O}(\cdot)$ denotes the complexity order, and $\mathbb{E}\{\cdot\}$ denotes the expectation operator with respect to all random variables.

\section{System model}
\label{sec2}

As illustrated in Fig.~\ref{fig:TRANS_RECIVE}, we consider a distributed wireless system comprising up to \(K\) potential emitters. The received signal, characterized by time-frequency overlap from concurrent emitter transmissions, which is denoted by
\begin{equation}
	y(t) \;=\; \sum_{m\in S} h_m \, d_m\!\left(x_m(t)\right) \;+\; w(t),
	\label{eq:rx_overview}
\end{equation}
where \(S\) is the  set of active devices; \(x_m(t)\) denotes the original transmitted signal of the $m$-th distributed emitter; \(d_m(\cdot)\) is the distortion function of the $m$-th emitter; \(h_m\) denotes the channel coefficient; \(w(t)\) is Gaussian noise.

Following the literature \cite{He2020CooperativeSEI}, we model the distortion function \(d_m(\cdot)\) using I/Q imbalance, spurious tone, carrier leakage, and power amplifier (PA) nonlinearity. Specifically,  the distorted signal with I/Q imbalance is expressed as
\begin{equation}
	x_m'(t)=\; \mu_m\, x_m(t) \;+\; \nu_m\, x_m^{*}(t),
	\label{eq:iq_model_new}
\end{equation}
where $\mu_m$ and $\nu_m$ are the parameters used to describe the distortion of the modulator, which can be represented as
\begin{align}
	\mu_m &= \tfrac{1}{2}\left(G_m+1\right)\cos\left(\tfrac{\zeta_m}{2}\right) + j\,\tfrac{1}{2}\left(G_m-1\right)\sin\left(\tfrac{\zeta_m}{2}\right), \label{eq:mu_m} \\
	\nu_m &= \tfrac{1}{2}\left(G_m-1\right)\cos\left(\tfrac{\zeta_m}{2}\right) + j\,\tfrac{1}{2}\left(G_m+1\right)\sin\left(\tfrac{\zeta_m}{2}\right), \label{eq:nu_m}
\end{align}
where $\zeta_m$ is the phase bias and $G_m$ denotes the gain imbalance. With the spurious tone and carrier leakage,  the distorted signal is further formulated by 
\begin{equation}
	x_m''(t) \;=\; \left(x_m'(t) + \xi_m\right) e^{j \left(2\pi f t\right)}\;+\; a_m^{\mathrm{ST}}\, e^{j\left(2\pi \left(f+f_m^{\mathrm{ST}}\right)t\right)},
	\label{eq:if_spur_cl_new}
\end{equation}
where \( a_m^{\mathrm{ST}} \) and \( f_m^{\mathrm{ST}} \) are the amplitude and the frequency of the spurious tone; \( \xi_m e^{j 2 \pi f t} \) is the carrier leakage.

Finally, the signal $x_m''(t)$ is fed into PA with the nonlinear distortion, which can be expressed as
\begin{align}
	x_m'''(t) &= \sum_{l=1}^{L} b_{m,l}\, \left( x_m''(t) \right)^{\,l}, \label{eq:pa_taylor_first} \\
	&= d_m\left(x_m(t)\right), \label{eq:pa_taylor_second}
\end{align}
where \( L \) is the order and \( b_{m,l} \) denotes the coefficient of the Taylor polynomial. The core objective of SMEI is to estimate the joint posterior probability of activation states for distributed emitters, which can be expressed as
\begin{equation}
	\hat{\mathbb{P}}\left(\boldsymbol{\lambda}|\bm{y}\right) = \mathcal{D}\left(\bm{y}\right),
	\label{eq:posterior_distribution}
\end{equation}
where
\begin{align}
	\bm{\lambda} &= \begin{bmatrix} \lambda_1, \lambda_2, \ldots, \lambda_K \end{bmatrix}^{\mathsf{T}}, 
	\label{eq:multi_label_vector} \\
	\bm{y} &= \left[y(1), y(2), \ldots, y(T)\right]^{\mathsf{T}}, 
	\label{eq:vector_y}
\end{align}
where \( \mathcal{D} \) represents the multi-emitter identifier to be designed, aiming to infer the activation states of these distributed emitters; the multi-label vector \( \boldsymbol{\lambda} \) is the emitter activation states, with \( \lambda_m \) being a binary variable that indicates the status of the \( m \)-th emitter (active if \( \lambda_m = 1 \) and inactive if \( \lambda_m = 0 \)); the vector \( \bm{y} \) represents the overlapping signal vector, where \( T \) denotes the length of the signal vector \( \bm{y} \).
\begin{remark}
	The core capability of SMEI lies in its ability to directly identify multiple concurrent emitters in distributed systems by leveraging their RF fingerprints. This technique enables the accurate determination of the set of active emitters from overlapping signals. In distributed systems, the uplink typically relies on simultaneous transmissions from multiple emitters, including sensor nodes. By deploying SMEI at the base station, this concurrent transmission characteristic can be fully utilized to authenticate the legitimate identities of distributed nodes participating in the computation, thereby enhancing the overall network's security.
\end{remark}

\begin{figure}[t]
	\centering
	\includegraphics[width=0.39\textwidth]{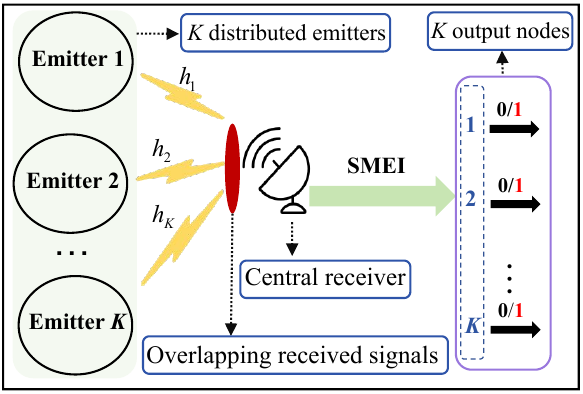}
	\caption{System model of SMEI with distributed emitters.}
	\label{fig:TRANS_RECIVE}
\end{figure}
\section{Theoretical Upper Bound of Identification Accuracy for SMEI}
\label{sec3}
In this section, we investigate the theoretical bounds for the SMEI problem. Specifically, we first present the performance metrics for SMEI. The theoretical limits of SMEI are then derived using Fano's inequality. Additionally, we estimate mutual information using mutual information neural estimation (MINE), which is used in the derivation of theoretical limits. These bounds provide fundamental performance guidelines for distributed emitter identification systems.

\subsection{Upper Bounds for SMEI}

We first define the metrics of subset accuracy and Hamming accuracy as
\begin{align}
	P_{\text{subset}} &= \frac{1}{N} \sum_{n=1}^{N} \mathcal{I}\left(\boldsymbol{\hat{\lambda}}^{(n)} = \boldsymbol{\lambda}^{(n)}\right), \\
	P_{\text{hamming}} &= \frac{1}{K} \sum_{k=1}^{K} \frac{1}{N} \sum_{n=1}^{N} \mathcal{I}\left(\hat{\lambda}_k^{(n)} = \lambda_k^{(n)}\right),
\end{align}
where $\bm{\hat{\lambda}}^{(n)}$ is the predicted label vector for the $n$-th sample; $\bm{\lambda}^{(n)}$ denotes the ground truth label vector for the $n$-th sample; $\hat{\lambda}_k^{(n)}$ is the predicted label for the $k$-th emitter of the $n$-th sample; $\lambda_k^{(n)}$ denotes the ground truth label for the $k$-th emitter of the $n$-th sample. Here, the subset accuracy is a strict metric that requires all labels to be predicted correctly for a sample to be counted as correct, while the Hamming accuracy is a more lenient metric that calculates the average accuracy across all individual label positions and is complementary to the Hamming loss \cite{DeSilva2025EVTC,Xu2024PerformanceEvaluation}.

The information-theoretic bounds for classification error rate are given in detail using Fano's inequality \cite{Brown2009InformationTheoretic,Zhou2010MultiInformation}. Using the information-theoretic framework, we rigorously derive theoretical upper bounds on identification accuracy for SMEI. The main challenges are handling the exponentially growing label-combination space when extending Fano's inequality from single-label to multi-label scenarios, and establishing the relationship between subset accuracy and Hamming accuracy in terms of conditional entropy. The principal results are formalized in the \textit{Lemma} \ref{lemma01}.
\begin{lemma}
	\label{lemma01}
	The theoretical upper bound of subset accuracy $P_{\text{subset}}$ and Hamming accuracy $P_{\text{hamming}}$ for SMEI are respectively given by
	\begin{align}
		\widetilde{P}_{\text{subset}} &= 
		\begin{cases}
			\frac{1}{|\bm{\varLambda}|}, & C \leq g(1 /|\bm{\varLambda}|) \\
			\min \left(1, g^{-1}(C)\right), & C > g(1 /|\bm{\varLambda}|) 
		\end{cases} \label{eq:subset_star_probability}, \\
		\widetilde{P}_{\text{hamming}} &= \left(\widetilde{P}_{\text{subset}}\right)^{\frac{1}{K}} \label{eq:hamming_star_probability},
	\end{align}
	where
	\begin{align}
		\bm{\varLambda} &= \left[\boldsymbol{\lambda}^{(1)}, \boldsymbol{\lambda}^{(2)}, \ldots, \boldsymbol{\lambda}^{(N)}\right], \\
		C &= \log\left(|\bm{\varLambda}| - 1\right) - H\left(\boldsymbol{\lambda}\right) + I\left(\boldsymbol{\lambda}; \boldsymbol{y}\right),\\
		g\left(C\right) &=  \log \left(|\bm{\varLambda}|-1\right) C+C \log \left(C\right) \\
		& \quad +\left(1-C\right) \log \left(1-C\right),
	\end{align}
	where $|\bm{\varLambda}| = 2^{K} - 1$ is the number of distinct label combinations in $\bm{\varLambda}$, excluding the empty transmitter case, and $\boldsymbol{y}$ denotes the received signal vector consisting of multiple samples following the model in \eqref{eq:rx_overview}. 
\end{lemma}
\textit{Proof:} See Appendix~\ref{appendix:fano_derivation}. \hfill $\square$
\subsection{Mutual Information Estimation via MINE}
\label{sec:mine}
The computation of the theoretical bounds $\widetilde{P}_{\text{subset}}$ and $\widetilde{P}_{\text{hamming}}$ in \textit{Lemma}~\ref{lemma01} requires estimating the mutual information $I\left(\bm{\lambda}; \bm{y}\right)$. To this end, we adopt the MINE framework proposed in~\cite{Belghazi2018MINE}, which efficiently estimates the mutual information between high-dimensional feature vectors and discrete label vectors by training a neural network estimator.

Practically, since $\bm{y}$ is high-dimensional, directly estimating $I\left(\bm{\lambda}; \bm{y}\right)$ is computationally infeasible. We therefore adopt a two-stage approach to estimate $I\left(\bm{\lambda}; \bm{y}\right)$.\footnote{For each SNR condition, the mutual information is estimated as a single scalar value, which is obtained through statistical aggregation over all test samples and reflects the overall dependency strength between the received signal and the multi-label vector.} In the first stage, we employ a deep residual convolutional neural network that has been trained on the training set to extract RF fingerprint feature vectors $\bm{x} \in \mathbb{R}^{q}$ from overlapping signal $\bm{y}$, where $q$ denotes the feature dimension. These extracted features $\bm{x}$ preserve key information for identifying multi-emitter combinations while significantly reducing data dimensionality. In the second stage, we utilize the MINE method to estimate the mutual information $I\left(\bm{\lambda}; \bm{x}\right)$ between these extracted features $\bm{x}$ and $\bm{\lambda}$, thereby indirectly obtaining an estimate of $I\left(\bm{\lambda}; \bm{y}\right)$.

Specifically, the MINE estimator learns a nonlinear function, which is represented by 
\begin{equation}
	T_{\bm{\theta}_{\text{MINE}}}: \mathbb{R}^{q} \times \{0,1\}^{K} \to \mathbb{R},
\end{equation}
where $\bm{\theta}_{\text{MINE}}$ represents the learnable parameters of the MINE network. For a sample pair $\left(\bm{x}, \bm{\lambda}\right)$, the function $T_{\bm{\theta}_{\text{MINE}}}$ outputs a real-valued scalar that quantifies the statistical dependence between them. The lower bound of mutual information is given by \cite{Belghazi2018MINE}
\begin{equation}
	\begin{aligned}
		I\left(\bm{\lambda}; \bm{x}\right) &\geq \sup_{\bm{\theta}_{\text{MINE}}} \mathbb{E}_{p\left(\bm{\lambda}, \bm{x}\right)} T_{\bm{\theta}_{\text{MINE}}}\left(\bm{\lambda}, \bm{x}\right) \\
		& \quad - \log \mathbb{E}_{p\left(\bm{\lambda}\right)p\left(\bm{x}\right)} \exp\left(T_{\bm{\theta}_{\text{MINE}}}\left(\bm{\lambda}, \bm{x}\right)\right).
	\end{aligned}
\end{equation}

During the training of the MINE network, we estimate the above expectations using the Monte Carlo method. For a batch of $N$ samples, the expectation under the joint distribution is estimated as
\begin{equation}
	\hat{\mathbb{E}}_{p\left(\bm{\lambda}, \bm{x}\right)} T_{\bm{\theta}_{\text{MINE}}}\left(\bm{\lambda}, \bm{x}\right) = \frac{1}{N} \sum_{n=1}^{N} T_{\bm{\theta}_{\text{MINE}}}\left(\bm{\lambda}^{(n)}, \bm{x}^{(n)}\right),
\end{equation}
and the expectation under the marginal distribution is computed by randomly shuffling the labels as
\begin{equation}
	\begin{aligned}
		&\hat{\mathbb{E}}_{p(\bm{\lambda})p(\bm{x})} \exp\left(T_{\bm{\theta}_{\text{MINE}}}\left(\bm{\lambda}, \bm{x}\right)\right) \\
		&= \frac{1}{N} \sum_{n=1}^{N} \exp\left(T_{\bm{\theta}_{\text{MINE}}}\left(\tilde{\bm{\lambda}}^{(n)}, \bm{x}^{(n)}\right)\right),
	\end{aligned}
\end{equation}
where $\tilde{\bm{\lambda}}$ denotes the shuffled label vector. Once the training is completed, we use the trained network to estimate mutual information between the $\bm{x}$ and $\bm{\lambda}$. By substituting the estimated mutual information into the bounds derived in \textit{Lemma}~\ref{lemma01}, we obtain quantitative theoretical upper bounds on the subset accuracy and Hamming accuracy for SMEI, which will be validated against experimental results in Section~\ref{sec6}.
\section{Multi-Label Learning Framework for SMEI}
\label{sec4}
In this section, we propose a SMEI approach for distributed emitters, which uses a multi-label learning framework to simultaneously identify multiple active emitters from overlapping signals. As illustrated in Fig.~\ref{Proposed_SMEI_nomp}, the proposed SMEI comprises two key components. The first is a multi-emitter fingerprint extractor $\xi\left(\cdot\right)$, which takes the preprocessed overlapping signal $\tilde{\bm{y}}$ as input and maps it to a $K$-dimensional fingerprint vector $\bm{\eta}$. The second is a multi-emitter activation detector $\sigma_\mathrm{s}\left(\cdot\right)$, which converts $\bm{\eta}$ into activation probabilities for identifying the active emitter set $\hat{S}^{*}$. In other words, in our SMEI approach, the joint posterior probability is estimated as
\begin{equation}
	\hat{\mathbb{P}}\left(\boldsymbol{\lambda}|\bm{y}\right) = \sigma_s\left(\xi\left(\tilde{\bm{y}}; \boldsymbol{\theta}_{\xi}\right)\right).
\end{equation}
\subsection{Multi-Emitter Fingerprint Extractor}
\label{sec4.1}
In the fingerprint extractor, we first preprocess $\bm{y}$ through energy normalization and global standardization, with real and imaginary parts as dual-channel input, yielding the preprocessed signal $\tilde{\bm{y}}$. Then, the multi-emitter fingerprint extractor $\xi\left(\cdot\right)$ transforms the preprocessed 
signal $\tilde{\bm{y}}$ into the multi-emitter fingerprint $\bm{\eta}$, 
which is a $K$-dimensional vector encoding hardware-induced characteristics of all potential emitters. This transformation is mathematically expressed by
\begin{equation}
	\label{eq:feature_extraction}
	\bm{\eta} = \xi\left(\tilde{\bm{y}}; \boldsymbol{\theta_{\xi}}\right),
\end{equation}
where $\boldsymbol{\theta_{\xi}}$ denotes the learnable parameters of the fingerprint extractor, 
and the element $\eta_m$ represents the activation score encoding the $m$-th emitter's fingerprint. 
\begin{figure*}[t]
	\centering
	\includegraphics[width=0.76\textwidth]{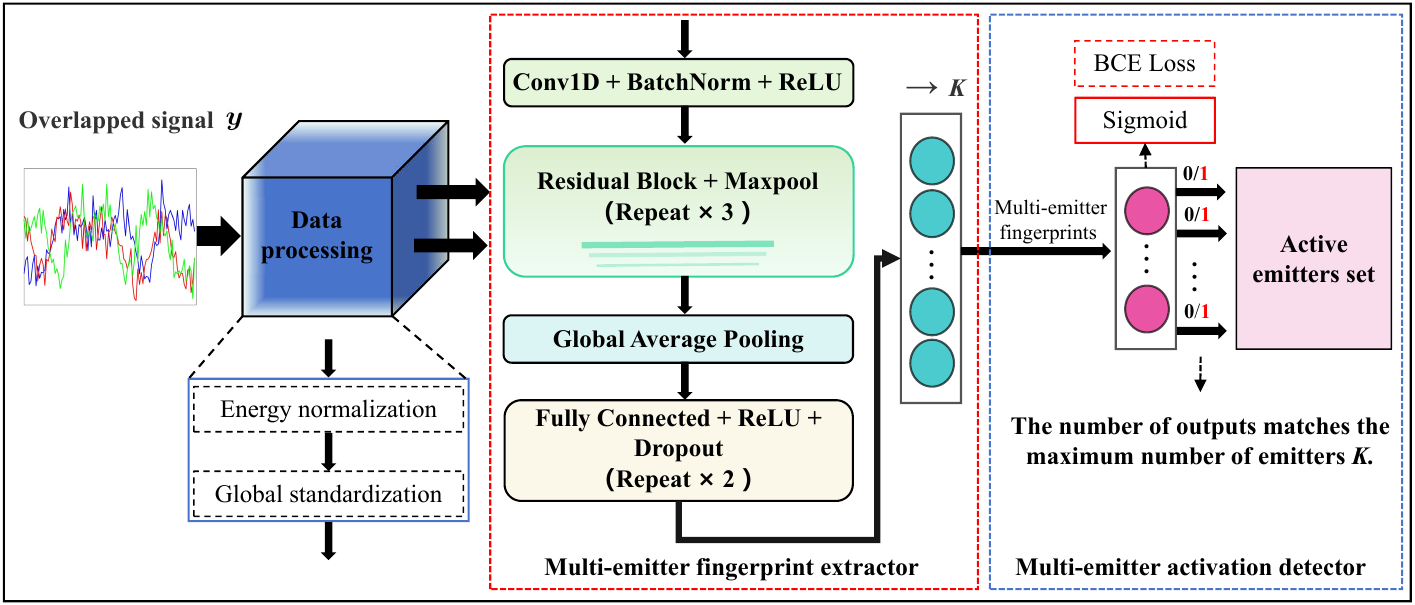}
	\caption{The network architecture of the proposed SMEI scheme which integrates a fingerprint extractor and an activation detector via multi-label learning.}
	\label{Proposed_SMEI_nomp}
\end{figure*}
\begin{algorithm}[t]
	\caption{Multi-emitter Fingerprint Extractor and Activation Detector for Solving Problem \eqref{eq:optimization}}
	\label{alg:smei_inference}
	
	\textbf{Input:} $\bm{y}$, $\boldsymbol{\theta_{\xi}}^*$, and $\tau$. \\
	\textbf{Output:} Active emitter set $\hat{S}^{*}$. \\
	1: Preprocess $\bm{y}$ through energy normalization and global standardization to obtain $\tilde{\bm{y}}$; \\
	2: Extract multi-emitter fingerprints using \eqref{eq:feature_extraction}; \\
	3: Compute identification probabilities using \eqref{eq:sigmoid_activation}; \\
	4: Determine active emitters using \eqref{eq:S_out}; \\
	5: \textbf{Return} $\hat{S}^{*}$.
\end{algorithm}
The function $\xi\left(\cdot\right)$ is implemented using a deep residual convolutional architecture~\cite{Sankhe2019ORACLE}, as detailed in Fig.~\ref{Proposed_SMEI_nomp}. Specifically, it consists of an initial convolutional layer, cascaded residual modules, global adaptive average pooling, fully connected layers, and a final linear classification layer. Learnable parameters $\boldsymbol{\theta_{\xi}}$ represent all weights and biases across these components. This network structure serves as one implementation example, and our framework does not depend on any specific architecture. Moreover, the multi-label learning framework is architecture-agnostic and can accommodate various backbone networks for different SMEI applications.

\subsection{Multi-emitter Activation Detector}
The activation detector converts the extracted multi-emitter fingerprint $\bm{\eta}$ 
into activation probabilities by applying a sigmoid function $\sigma_s\left(\cdot\right)$. Specifically, the probability of the $m$-th emitter is computed by
\begin{equation}
	\label{eq:sigmoid_activation}
	\begin{aligned}
		p_m &= \sigma_\mathrm{s}\left(\eta_m\right), \\
		&= \frac{1}{1 + \exp\left(-\eta_m\right)}.
	\end{aligned}
\end{equation}

Unlike softmax-based multiclass schemes~\cite{Sankhe2019ORACLE}, $\sigma_\mathrm{s}\left(\cdot\right)$ enables independent probability estimation for each emitter's fingerprint, 
allowing multiple concurrent activations. To learn the optimal parameters $\boldsymbol{\theta}_{\xi}^{*}$ for the extractor, we minimize the binary cross-entropy (BCE) loss \cite{Wu2024Transformer} between the network's predictions and the ground truth multi-label vectors $\boldsymbol{\lambda}$. Thus, the optimization problem is formulated as
\begin{equation}
	\label{eq:optimization}
	\boldsymbol{\theta}_{\xi}^{*} = \underset{\boldsymbol{\theta}_{\xi}}{\arg\min} \,\mathcal{L}_{\,\text{BCE}}\left(\bm{p}_{\xi}, \boldsymbol{\lambda}\right),
\end{equation}
where $\bm{p}_{\xi} = \sigma_s\left(\xi\left(\tilde{\bm{y}}; \boldsymbol{\theta}_{\xi}\right)\right)$ denotes the predicted activation probabilities, and the BCE loss\protect\footnotemark is defined as
\begin{equation}
	\label{eq:bce_loss}
	\!\!\!\mathcal{L}_{\,\text{BCE}} \!= -\frac{1}{K}\!\!\sum_{m=1}^{K} \left(\lambda_m \log p_m \!+\! \left(1-\lambda_m\right) \log\left(1-p_m\right) \right),
\end{equation}
where $p_m$ denotes the predicted activation probability of the $m$-th emitter.
\footnotetext{In implementation, the BCE loss is computed directly from logits $\bm{\eta}$ using the log-sum-exp trick to ensure numerical stability \cite{paszke2019pytorch}.}

Finally, the multi-emitter activation decision is obtained by applying a thresholding criterion, which can be expressed by
\begin{equation}\label{eq:S_out}
	\hat{S}^{*} = \left\{ m \mid p_m^{*} > \tau, m \in \left[1, K\right] \right\},
\end{equation}
where $\tau \in \left(0,1\right)$ denotes the decision threshold. The complete identification procedure is summarized in Algorithm~\ref{alg:smei_inference}.

\begin{remark}
We propose an SMEI using multi-label learning to identify overlapping RF signals, representing a paradigm shift from traditional single-label multiclass classification methods. The multi-label approach enables simultaneous identification of multiple active emitters without requiring explicit signal separation or prior knowledge of the number of active emitters. Moreover, it employs independent binary classification for each potential emitter, naturally aligning with the practical scenario in which distributed emitters operate independently and may be active simultaneously.
\end{remark}

\section{Message-passing Enhancement for SMEI}
\label{sec5}
In this section, we propose the I-SMEI approach for distributed emitters, which enhances the SMEI framework by leveraging message-passing mechanisms to better exploit inter-emitter correlations in overlapping signals. We first establish the theoretical motivation by demonstrating that significant correlations exist between signals sharing common emitters. Building on this motivation, we subsequently detail the implementation of message passing via multi-head attention, which enables the network to effectively capture and leverage these inter-emitter dependencies at the feature level to improve fingerprint disentanglement.
\subsection{Correlation-driven Probabilistic Reasoning}
In distributed emitters scenarios, the multi-emitter fingerprints are intertwined in the received signals, leading to complex correlations among emitters. Leveraging these correlations can reduce identification uncertainty. We formalize this observation in the following \textit{Lemma} \ref{lemma02}.
\begin{lemma}
	\label{lemma02}
	For any two label vectors $\boldsymbol{\lambda}^{\left(i\right)}, \boldsymbol{\lambda}^{\left(j\right)} \in \left\{0,1\right\}^K$, we define their corresponding received signals as $\boldsymbol{y}^{\left(i\right)}$ and $\boldsymbol{y}^{\left(j\right)}$, respectively. If the intersection of the label vectors is non-empty, i.e. $\boldsymbol{\lambda}^{\left(i\right)} \odot \boldsymbol{\lambda}^{\left(j\right)} \neq \boldsymbol{0}$, then there exists a non-zero correlation between the corresponding received signals, i.e.\footnote{The conditional mutual information is defined as $I(X; Y \mid Z) = H(X \mid Z) - H(X \mid Y, Z)$~\cite{cover2006elements}.}
	\begin{equation}
		I\left(\boldsymbol{y}^{\left(i\right)}; \boldsymbol{y}^{\left(j\right)} \mid \boldsymbol{\lambda}^{\left(i\right)}, \boldsymbol{\lambda}^{\left(j\right)}\right) > 0.
	\end{equation}
\end{lemma}
\begin{proof}
	According to the system model (Section~\ref{sec2}, Eq.\eqref{eq:rx_overview}), the received signal follows a linear superposition model as
	\begin{equation}
		\boldsymbol{y} = \sum_{k=1}^{K} \lambda_k \boldsymbol{s}_k\left(\boldsymbol{\theta}_k\right) + \boldsymbol{w},
	\end{equation}
	where $\boldsymbol{s}_k\left(\boldsymbol{\theta}_k\right)$ is the signal component of the $k$-th emitter carrying RF fingerprint parameter $\boldsymbol{\theta}_k$, and $\boldsymbol{w}$ is the additive noise. For the received signals corresponding to label vectors $\boldsymbol{\lambda}^{\left(i\right)}$ and $\boldsymbol{\lambda}^{\left(j\right)}$, they can be decomposed by
\begin{align}
	\boldsymbol{y}^{\left(i\right)} &= \sum_{k \in \mathcal{A}_i} \boldsymbol{s}_k\left(\boldsymbol{\theta}_k\right) + \boldsymbol{w}_i,\\
	\boldsymbol{y}^{\left(j\right)} &= \sum_{k \in \mathcal{A}_j} \boldsymbol{s}_k\left(\boldsymbol{\theta}_k\right) + \boldsymbol{w}_j,
	\label{eq:second}
\end{align}
where $\mathcal{A}_i = \left\{k: \lambda_k^{\left(i\right)}=1\right\}$ and $\mathcal{A}_j = \left\{k: \lambda_k^{\left(j\right)}=1\right\}$ are the index sets of active transmitters corresponding to the two label vectors, respectively. If $\mathcal{A}_i \cap \mathcal{A}_j \neq \emptyset$, denote the shared transmitter set as $\mathcal{C} = \mathcal{A}_i \cap \mathcal{A}_j$, then both received signals contain the common signal component $\sum_{k \in \mathcal{C}} \boldsymbol{s}_k\left(\boldsymbol{\theta}_k\right)$. The conditional mutual information is given as
\begin{equation}
	\begin{aligned}
		I\left(\boldsymbol{y}^{(i)}; \boldsymbol{y}^{(j)} \!\mid \! \boldsymbol{\lambda}^{(i)}, \boldsymbol{\lambda}^{(j)}\right) 
		&\!=\! H\left(\!\boldsymbol{y}^{(i)} \mid \boldsymbol{\lambda}^{(i)}\!\right) \\
		&\quad \!\!- H\left(\!\boldsymbol{y}^{(i)} \mid \boldsymbol{y}^{(j)}, \boldsymbol{\lambda}^{(i)}, \boldsymbol{\lambda}^{(j)}\!\right).
	\end{aligned}
	\label{eq:mutual_information}
\end{equation}

The shared signal allows $\boldsymbol{y}^{\left(j\right)}$ to provide information about the shared emitter fingerprint parameters $\left\{\boldsymbol{\theta}_k\right\}_{k \in \mathcal{C}}$, thereby reducing the conditional uncertainty of $\boldsymbol{y}^{\left(i\right)}$. According to the monotonicity theorem of conditional entropy \cite{cover2006elements}, $H(X|Y,Z) \leq H(X|Z)$ with equality if and only if $X$ and $Y$ are conditionally independent given $Z$. Since $\boldsymbol{y}^{\left(i\right)}$ and $\boldsymbol{y}^{\left(j\right)}$ are not conditionally independent given $\boldsymbol{\lambda}^{\left(i\right)}$ and $\boldsymbol{\lambda}^{\left(j\right)}$ due to the shared signal components, i.e.
\begin{align}
	H\left(\boldsymbol{y}^{\left(i\right)} | \boldsymbol{y}^{\left(j\right)}, \boldsymbol{\lambda}^{\left(i\right)}, \boldsymbol{\lambda}^{\left(j\right)}\right) & < H\left(\boldsymbol{y}^{\left(i\right)} | \boldsymbol{\lambda}^{\left(i\right)}\right).
	\label{fig:monotonicity}
\end{align}

Then we substitute~\eqref{fig:monotonicity} into the definition of conditional mutual information in~\eqref{eq:mutual_information}, i.e.  
\begin{align}
	I\left(\boldsymbol{y}^{\left(i\right)}; \boldsymbol{y}^{\left(j\right)} | \boldsymbol{\lambda}^{\left(i\right)}, \boldsymbol{\lambda}^{\left(j\right)}\right) & > 0.
\end{align}

This completes proof of \textit{Lemma}~\ref{lemma02}.
\end{proof}

\begin{figure}[t]
	\centering
	\includegraphics[width=0.83\columnwidth]{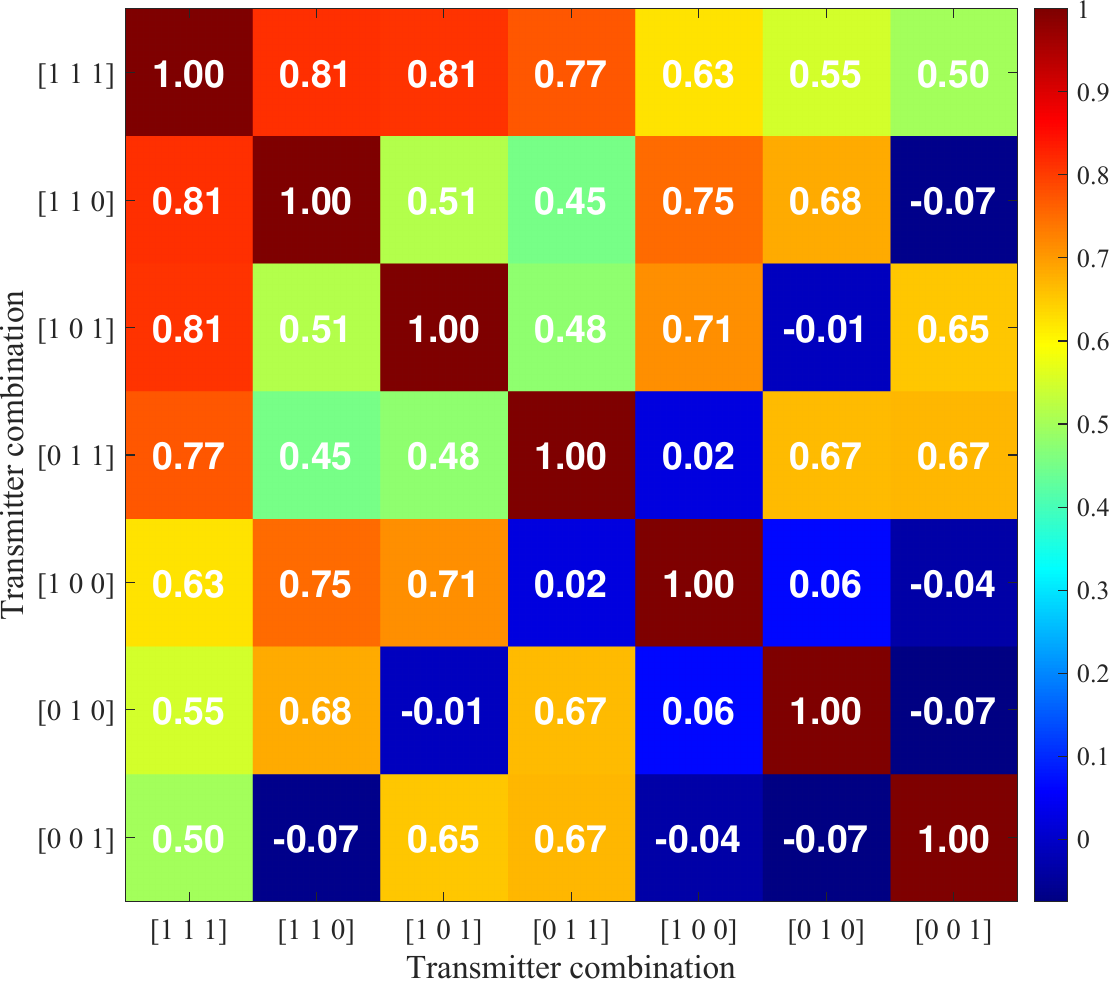}
	\caption{
		The Pearson correlation\protect\footnotemark matrix of the average time-domain waveforms for seven signal classes (300 samples per class)  under an overlap = 50\% scenario ($K=3$, AWGN channel, and SNR = 18 dB). When distinct signal combinations share emitters, the correlation between their waveforms provides a foundation for a message-passing mechanism.}
	\label{fig:heatmap_corr}
\end{figure}
\footnotetext{Signals are preprocessed via average power normalization and separate I/Q z-scoring. For signal classes $y^a$ and $y^b$ with average time-domain waveforms $\bar{y}^a$ and $\bar{y}^b$, the Pearson correlation is $\rho_{ab} = \mathbb{E}[(\bar{y}^a - \mu_a)(\bar{y}^b - \mu_b)]/(\varOmega_a \varOmega_b)$, where $\mu_a$, $\mu_b$ and $\varOmega_a$, $\varOmega_b$ are the respective means and standard deviations.}
\begin{figure}[]
	\centering
	\includegraphics[width=0.43\textwidth]{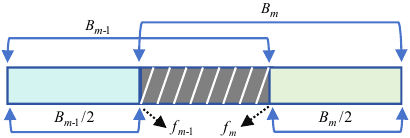}
	\caption{The schematic of ``overlap = 50\%'' between adjacent equal-bandwidth bands $B_{m-1} = B_{m}$. The bandwidths $B_{m-1}$ and $B_{m}$ of the $(m-1)$-th and $m$-th emitters represent equal bandwidths, with $B_{m-1}/2$ and $B_{m}/2$ indicating the half bandwidths of these two emitters, while $f_{m-1}$ and $f_{m}$ represent the center frequencies of these two bands, respectively. }
	\label{Bandover}
\end{figure}

To empirically validate \textit{Lemma} \ref{lemma02}, Fig.~\ref{fig:heatmap_corr} demonstrates the phenomenon by presenting the Pearson correlation matrix between overlapping signals under different emitter combinations with 50\% spectrum overlap (as illustrated in Fig.~\ref{Bandover}). The correlation matrix exhibits clear, diverse correlation structures, with signals from common emitters showing stronger correlations. This visual evidence indicates that the identification state of any individual emitter is inherently influenced by the presence and RF characteristics of coexisting emitters, thereby motivating the exploitation of inter-emitter correlations to enhance fingerprint disentanglement.
\begin{figure*}[t]
	\centering
	\includegraphics[width=0.9\textwidth]{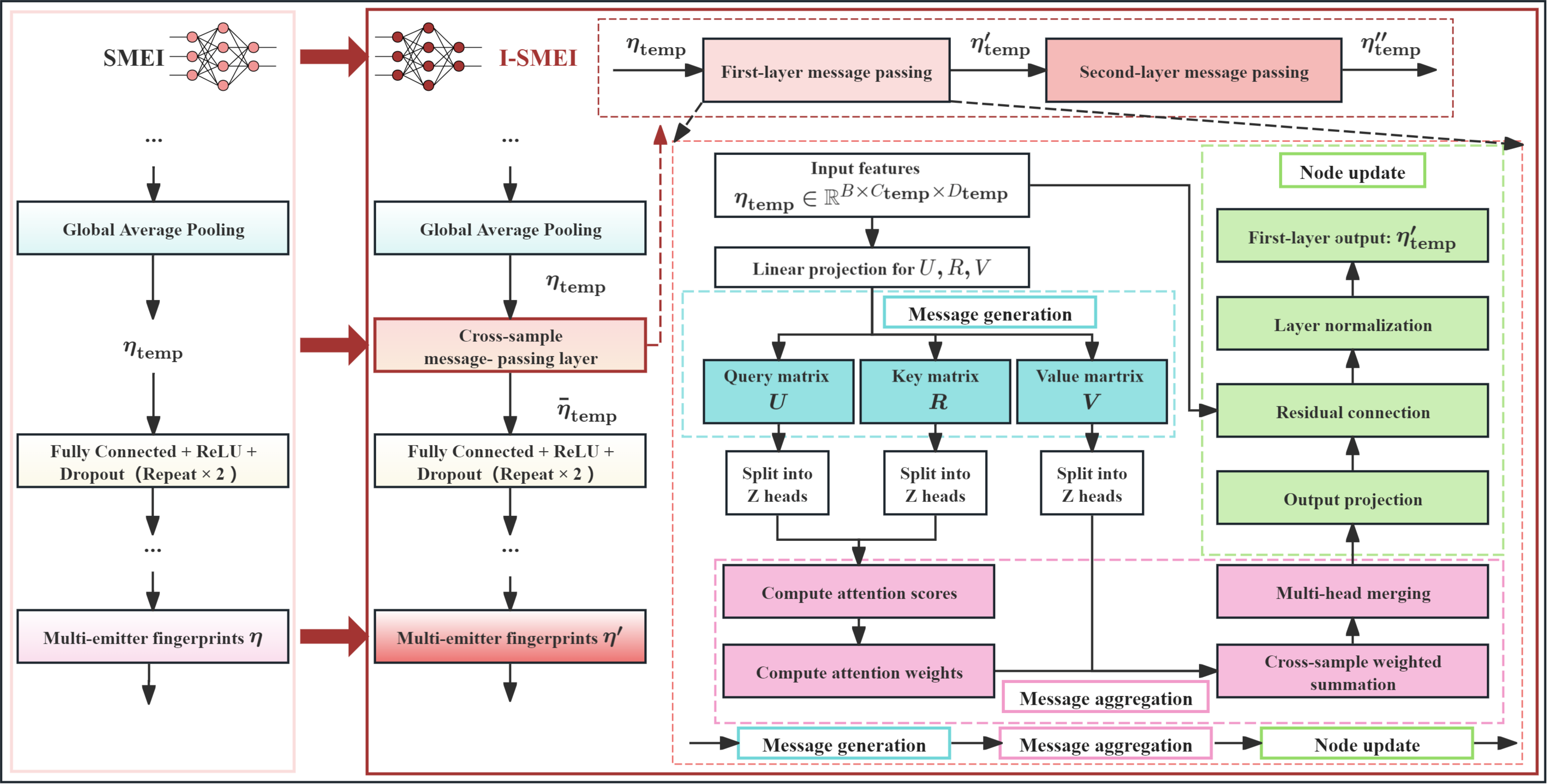}
	\caption{The message-passing layer structure of the proposed I-SMEI which employs multi-head attention mechanism following the three-step paradigm of message generation, message aggregation, and node update. }
	\label{Proposed_imSMEI_nomp}
\end{figure*}

Motivated by this observation, the core idea of I-SMEI is to enable mutual interaction among fingerprint feature representations of different emitters via a message-passing mechanism while maintaining computational efficiency, thereby implicitly capturing statistical correlations among emitters in the feature space. Specifically, we design a feature-level information-exchange mechanism that enables the RF fingerprint features corresponding to each emitter to acquire relevant information from other emitters' fingerprint features when computing the activation probability of that emitter. Mathematically, our I-SMEI framework enhances the proposed SMEI by introducing a message-passing mechanism to refine the fingerprint representation. The joint posterior  probability is estimated as
\begin{equation}
	\hat{\mathbb{P}}\left(\boldsymbol{\lambda}|\bm{y}\right) = \sigma_s\left(\bm{\eta}'; \boldsymbol{\theta}_{\varPhi}, \boldsymbol{\theta}_{\xi}\right),
\end{equation}
where $\bm{\eta}'$ represents the enhanced multi-emitter fingerprint obtained by applying the message-passing function $\varPhi(\cdot; \boldsymbol{\theta}_{\varPhi})$, and $\boldsymbol{\theta}_{\varPhi}$ denotes the learnable parameters of the message-passing function. The detailed implementation of $\varPhi$ is presented in the following subsection.

The tensor $\bar{\bm{\eta}}_{\mathrm{temp}}$ serves as an intermediate feature representation that has been further refined through message passing. Subsequently, we continue the standard processing pipeline, which excludes the message-passing layer, to extract multi-emitter fingerprints, ultimately yielding the final multi-emitter fingerprints $\bm{\eta}'$.
\subsection{Message passing via Multi-head Attention}
We formulate the message-passing process following the established neural message-passing paradigm \cite{NeuralMessagePassing2017}, which comprises three fundamental steps of message generation, message aggregation, and node update. Here, the message-passing function $\varPhi(\cdot)$ is implemented via a two-layer process of cross-sample information transfer that systematically follows the three-step paradigm. This mechanism efficiently captures inter-emitter correlations within the RF fingerprint features, as visualized in Fig.~\ref{Proposed_imSMEI_nomp}. Mathematically, the message-passing function is expressed as
\begin{equation} \label{eq:message_passing_function}
	\bar{\bm{\eta}}_{\mathrm{temp}} = \varPhi\left(\bm{\eta}_{\mathrm{temp}}; \boldsymbol{\theta}_{\varPhi}\right),
\end{equation}
where $\bm{\eta}_{\mathrm{temp}} \in \mathbb{R}^{B \times C_{\mathrm{temp}} \times D_{\mathrm{temp}}}$ represents the input of the message-passing process, with the output $\bar{\bm{\eta}}_{\mathrm{temp}} \in \mathbb{R}^{B \times C_{\mathrm{temp}} \times D_{\mathrm{temp}}}$ having the same dimensions;\footnote{$\bm{\eta}_{\mathrm{temp}}$ represents our intermediate feature tensor, distinct from the final fingerprint vector $\bm{\eta}$, as defined in \eqref{eq:feature_extraction}.} \(B\) represents the number of samples of the received overlapping signals; \(C_{\mathrm{temp}}\) is the number of features; \(D_{\mathrm{temp}}\) denotes the dimension of the features. The parameter set $\boldsymbol{\theta}_{\varPhi}$ is the collection of all learnable parameters of the message-passing function, which is formally defined by
\begin{equation} \label{eq:learnable_params}
	\boldsymbol{\theta}_{\varPhi} = \left\{ \bm{W}_{U}, \bm{W}_{R}, \bm{W}_{V}, \bm{W}_{O} \right\},
\end{equation}
where $\bm{W}_{U}$, $\bm{W}_{R}$, and $\bm{W}_{V}$ denote the weight matrices for message generation of query, key, and value vectors respectively; $\bm{W}_{O}$ is the output weight matrix for node update. For each attention head $z$, we partition the queries, keys, and values into $Z$ heads, which can be expressed as
\begin{align}
	\bm{U} &= [\bm{U}_1, \bm{U}_2, \ldots, \bm{U}_Z], \label{eq:U} \\
	\bm{R} &= [\bm{R}_1, \bm{R}_2, \ldots, \bm{R}_Z], \label{eq:R} \\
	\bm{V} &= [\bm{V}_1, \bm{V}_2, \ldots, \bm{V}_Z], \label{eq:V}
\end{align}
where $\bm{U}_z$, $\bm{R}_z$, and $\bm{V}_z$ denote the query, key, and value matrices for the $z$-th head, respectively. More precisely, these matrices are represented by
\begin{align}
	\bm{U}_{z} &= \bm{W}_{U,z}\bm{\eta}_{\mathrm{temp}}, \label{eq:U_z} \\
	\bm{R}_{z} &= \bm{W}_{R,z}\bm{\eta}_{\mathrm{temp}}, \label{eq:R_z} \\
	\bm{V}_{z} &= \bm{W}_{V,z}\bm{\eta}_{\mathrm{temp}}, \label{eq:V_z}
\end{align}
where $\bm{W}_{U,z}$, $\bm{W}_{R,z}$, and $\bm{W}_{V,z}$ are the weight matrices used to derive the query, key, and value matrices for the $z$-th head from the input features $\bm{\eta}_{\mathrm{temp}}$.

Next, during  the message aggregation phase, the attention scores for each attention head are computed as
\begin{equation} \label{eq:attention_scores}
	\bm{A}_{z}[i,j] = \varsigma\left(\frac{\bm{U}_{z,i}\bm{R}_{z,j}^{\top}}{\sqrt{D_{\mathrm{temp}}/Z}}\right),
\end{equation}
where $i,j=1,\dots,B$, $\bm{A}_{z} \in \mathbb{R}^{B \times B}$ is the weight matrix associated with the $z$-th attention head, and $\varsigma(\cdot)$ denotes the softmax normalization function, indicating that $\bm{A}_{z}[i,j]$ represents the correlation between the $i$-th overlapping signal and the $j$-th overlapping signal.\footnote{Although our scheme uses multiple attention heads, it operates across samples rather than within any single sample, and thus differs from the standard multi-head self-attention mechanism.} 

Following message aggregation, node update is performed on the fingerprint features from the $i$-th overlapping signal, which can be calculated by
\begin{equation} \label{eq:updated_features}
	\bm{\eta}_{\mathrm{temp},i}' = \mathcal{R}\left(\bm{\eta}_{\mathrm{temp},i} + \sum_{z=1}^{Z} \sum_{j=1}^{B} \bm{A}_{z}[i,j] \bm{V}_{z,j}\right),
\end{equation}
where $\mathcal{R}(\cdot)$ denotes the layer normalization function, and ~\eqref{eq:updated_features} indicates that the new representation of each overlapping signal is a weighted aggregation of its own fingerprint features and those of other overlapping signals, with weights dynamically determined by inter-signal information transfer to achieve efficient information integration.

\begin{algorithm}[t]
	\caption{Message-Passing Algorithm via Multi-head Attention}
	\label{alg:message_passing}
	\textbf{Input:} $\bm{\eta}_{\mathrm{temp}}$, $\boldsymbol{\theta}_{\varPhi}$, and $Z$. \\
	\textbf{Output:} $\bar{\bm{\eta}}_{\mathrm{temp}}$. \\
	1: Compute $\bm{U}, \bm{R}$, and $\bm{V}$ from using \eqref{eq:U}, \eqref{eq:R}, and \eqref{eq:V}; \\
	2: \textbf{for} $z = 1$ to $Z$ \textbf{do} \\
	3: \ \ Compute $\bm{U}_{z}, \bm{R}_{z}$, and $\bm{V}_{z}$ using \eqref{eq:U_z}, \eqref{eq:R_z}, and \eqref{eq:V_z}; \\
	4: \ \ Compute $\bm{A}_{z}[i,j]$ using \eqref{eq:attention_scores}; \\
	5: \textbf{end for} \\
	6: \textbf{for} $i = 1$ to $B$ \textbf{do} \\
	7: \ \ Update fingerprint features $\bm{\eta}_{\mathrm{temp},i}'$ using \eqref{eq:updated_features}; \\
	8: \textbf{end for} \\
	9: Combine features $\bm{\eta}_{\mathrm{temp}}'$ using \eqref{eq:residual_connection}; \\
	10: Compute final representation $\bm{\eta}_{\mathrm{temp}}''$ using \eqref{eq:final_representation}; \\
	11: \textbf{Return} $\bar{\bm{\eta}}_{\mathrm{temp}} = \bm{\eta}_{\mathrm{temp}}''$.
\end{algorithm}

The $\bm{\eta}_{\mathrm{temp}}$ is then combined with the updated fingerprint features through a residual connection, which is given by
\begin{equation} \label{eq:residual_connection}
	\bm{\eta}_{\mathrm{temp}}' = \mathcal{R}\left(\bm{\eta}_{\mathrm{temp}} + \bm{W}_{O}\left(\sum_{z=1}^{Z} \bm{A}_{z} \bm{V}_{z}\right)\right).
\end{equation}

Consequently, the final feature representation $\bm{\eta}_{\mathrm{temp}}''$ after the second round of information transfer can be expressed as 
\begin{align}
	\bm{\eta}_{\mathrm{temp}}'' &= \mathcal{R}\left(\bm{\eta}_{\mathrm{temp}}' + \bm{W}_{O}'\left(\sum_{z=1}^{Z} \left(\bm{A}_{z}' \bm{V}_{z}'\right)\right)\right), \label{eq:final_representation} \\
	\bar{\bm{\eta}}_{\mathrm{temp}} &= \bm{\eta}_{\mathrm{temp}}'',
\end{align}
where indicates that $\bar{\bm{\eta}}_{\mathrm{temp}}$ is the result of two rounds of iterative information transfer, thereby refining the output of the message-passing function. Note that the number of iterations in message passing does not need to be limited to the two rounds considered in this paper, and more rounds can be performed. The complete procedure for $\varPhi(\cdot)$ is detailed in Algorithm~\ref{alg:message_passing}.

We place the message-passing module after global average pooling and before the fully connected layer because global average pooling transforms time-frequency features into a compact feature tensor, providing a clear representation of the input overlapping signals. This transformation is essential, as it allows message passing to perform weighted integration of the features from distributed emitters within the batch, explicitly modeling the inter-correlations among them. Here, $\bar{\bm{\eta}}_{\mathrm{temp}}$ is still an intermediate feature tensor. Subsequently, we continue with the parts of the standard processing flow that do not include the message-passing layer (see Section~\ref{sec4.1}) to perform multi-emitter fingerprint extraction, ultimately obtaining the feature $\bm{\eta}'$. Immediately, the joint optimization problem for the I-SMEI framework integrates the message-passing function can be expressed as
\begin{equation}\label{eq:joint_optimization_ismei}
	\left(\boldsymbol{\theta}_{\xi}^*, \boldsymbol{\theta}_{\varPhi}^*\right) 
	= \underset{\boldsymbol{\theta}_{\xi}, \boldsymbol{\theta}_{\varPhi}}{\arg \min}  \,\mathcal{L}_{\mathrm{BCE}}\left(\hat{\bm{p}}_{\varPhi}, \boldsymbol{\lambda}\right),
\end{equation}
where $\hat{\bm{p}}_{\varPhi} = [\hat{p}_1, \hat{p}_2, \ldots, \hat{p}_K]^{\mathsf{T}}$ denotes the predicted activation probabilities with message-passing enhancement; $\hat{\mathbb{P}}\left(\boldsymbol{\lambda}|\bm{y}\right) = \sigma_s\left(\bm{\eta}'; \boldsymbol{\theta}_{\varPhi}, \boldsymbol{\theta}_{\xi}\right)$ is the predicted probability for the $m$-th emitter within message passing; $\mathcal{L}_{\mathrm{BCE}}$ is defined as in \eqref{eq:bce_loss}. 
\begin{remark}
The I-SMEI method employs a message-passing mechanism to dynamically aggregate RF fingerprints from overlapping signals, capturing the correlations among distributed emitters. Specifically, through message passing, each emitter can adaptively integrate feature information from neighboring sources to achieve mutual feature-level enhancement, thereby improving the identification capability.
\end{remark}

\setcounter{table}{1}
\begin{table*}[t]
	\centering
	\caption{\textsc{Transmitter parameters}}
	\label{tab:device-params}
	\begin{tabular}{lcccccc}
		\toprule
		Device & \(G\) & \phantom{-}\(\zeta\) \(\cdot\frac{\pi}{180}\) & \( a^{\mathrm{ST}} \) & \( f^{\mathrm{ST}} \) \textmd{(MHz)} & \(\xi_m \times 10^{-3}\) & [$b_1$,$b_2$,$b_3$] \\
		\midrule
		Dev1 & 0.9998 & -0.0180\ & 0.0082 & 0.129 & 1.3 + 8.2\(j\)  & [1.00, 0.50, 0.30] \\
		Dev2 & 1.0056 & \phantom{-}0.0175\ & 0.0075 & 0.132 & 1.5 + 7.2\(j\) & [1.00, 0.08, 0.60] \\
		Dev3 & 1.0102 & \phantom{-}0.0120\ & 0.0070 & 0.123 & 1.1 + 6.8\(j\) & [1.00, 0.01, 0.01] \\
		Dev4 & 0.9992 & \phantom{-}0.0030\ & 0.0087 & 0.135 & 1.7 + 9.0\(j\) & [1.00, 0.01, 0.40] \\
		Dev5 & 0.9982 & \phantom{-}0.0240\ & 0.0090 & 0.119 & 2.0 + 6.5\(j\) & [1.00, 0.60, 0.08] \\
		\bottomrule
	\end{tabular}
\end{table*}
\setcounter{table}{2}
\begin{table*}[]
	\centering
	\caption{\textsc{Digital IF configurations under different spectral overlap levels ($K=3$)}}
	\label{tab:dif-config}
	\begin{tabular}{lccc}
		\toprule
		Overlap/(GHz) & 0\% & 50\%  & 100\% \\
		\midrule
		Configuration 1 & [2.4140, 2.4400, 2.4660] & [2.4270, 2.4400, 2.4530] & [2.4270, 2.4270, 2.4270] \\
		Configuration 2 & [2.4140, 2.4660, 2.4400] & [2.4270, 2.4530, 2.4400] & [2.4530, 2.4530, 2.4530] \\
		Configuration 3 & [2.4400, 2.4140, 2.4660] & [2.4400, 2.4270, 2.4530] & [2.4400, 2.4400, 2.4400] \\
		Configuration 4 & [2.4400, 2.4660, 2.4140] & [2.4400, 2.4530, 2.4270] & [2.4140, 2.4140, 2.4140] \\
		Configuration 5 & [2.4660, 2.4140, 2.4400] & [2.4530, 2.4270, 2.4400] & [2.4660, 2.4660, 2.4660] \\
		Configuration 6 & [2.4660, 2.4400, 2.4140] & [2.4530, 2.4400, 2.4270] & --- \\
		\bottomrule
	\end{tabular}
\end{table*}
\setcounter{table}{0}
\begin{table}[]
	\centering
	\caption{\textsc{System Configuration Parameters in Multi-Emitter Identification Experiments}}
	\label{tab:combined}
	\resizebox{\linewidth}{!}{%
		\begin{tabular}{lll}
			\toprule
			Category & Parameter & Value \\
			\midrule
			\multirow{6}{*}{Signal parameters} 
			& Sampling rate & 120 MHz \\
			& Symbol rate & 20 MHz \\
			& Oversampling ratio & 6 \\
			& RRC roll-off & \(\alpha\) = 0.3 \\
			& RRC span & \(10 \) (symbols) \\
			& Rician factor & 10 dB \\
			\midrule
			\multirow{8}{*}{Training hyperparameters} 
			& Batch size & 128 \\
			& Epochs & 100 \\
			& Initial learning rate & $5 \times 10^{-4}$ \\
			& Step size & 25 \\
			& Learning rate decay factor & 0.5 \\
			& Warmup epochs & 5 \\
			& Warmup ratio & 0.1 \\
			& Early stopping patience & 25 epochs \\
			\midrule
			\multirow{3}{*}{Optimizer}
			& Type & Adam \\
			& Base learning rate & $5 \times 10^{-4}$ \\
			& Learning rate scheduler & Linear warmup + Step decay \\
			\bottomrule
		\end{tabular}%
	}
\end{table}

\section{PERFORMANCE ANALYSIS}  
In this section, we first introduce the experimental setup and baseline approaches. Following that, the performance of the proposed SMEI and I-SMEI designs is thoroughly assessed in terms of identification accuracy, computational complexity, and robustness, supported by comprehensive numerical results and ablation analyses.
\label{sec6}
\subsection{Simulation Setup and Baselines}
This section describes the experimental setup for dataset synthesis, including signal modulation, spectrum allocation, and channel models. We use quadrature phase-shift keying for drone communication in the 2.40-2.48 GHz range \cite{Yu2024DroneRFa}. To avoid boundary transients, we extract input signals from the steady-state portion of each sequence. Particularly, we consider three levels of spectrum overlap: 0\%, 50\%, and 100\%. In the ``overlap = 50\%'' configuration, there is a 13 MHz separation between adjacent carrier frequencies. The training, validation, and test datasets are split 4:1:1, with 1200, 300, and 300 samples per emitter, respectively. Detailed system configuration parameters for the multi-emitter identification trials are provided in Table~\ref{tab:combined}, with the specific device parameters and digital intermediate frequency configurations detailed in Table~\ref{tab:device-params} and Table~\ref{tab:dif-config}, respectively.

Subsequently, we conduct ablation experiments to systematically evaluate the proposed methods. First, we compare SMEI with the multiclass classification scheme from \cite{Sankhe2019ORACLE} in terms of identification accuracy, parameter count, and computational complexity. Second, we compare SMEI with a multi-label approach using STFT feature extraction as described in \cite{Yu2024DroneRFa}. Third, for enhanced schemes, we compare I-SMEI with the Q2L scheme \cite{Liu2021Query2Label}, both of which incorporate feature enhancement mechanisms into the multi-label framework. The Q2L scheme uses a single Transformer decoder layer with 4 attention heads. All methods are evaluated under unified experimental conditions. For each method, we select and save the best-performing model via the training metrics $P_{\text{subset}}$ and $P_{\text{hamming}}$, and both metrics are reported during testing.

\subsection{Numerical Results}
\begin{observation}
The proposed SMEI scheme achieves comparable identification accuracy to that of existing multiclass methods while significantly reducing model complexity as $K$ increases. (cf. Figs.~\ref{fig:pc_subset_allk}, \ref{fig:params_allk}, and Table~\ref{tab:complexity-comparison})
\end{observation}

\begin{figure}[t]
	\centering
	\includegraphics[width=0.8\linewidth]{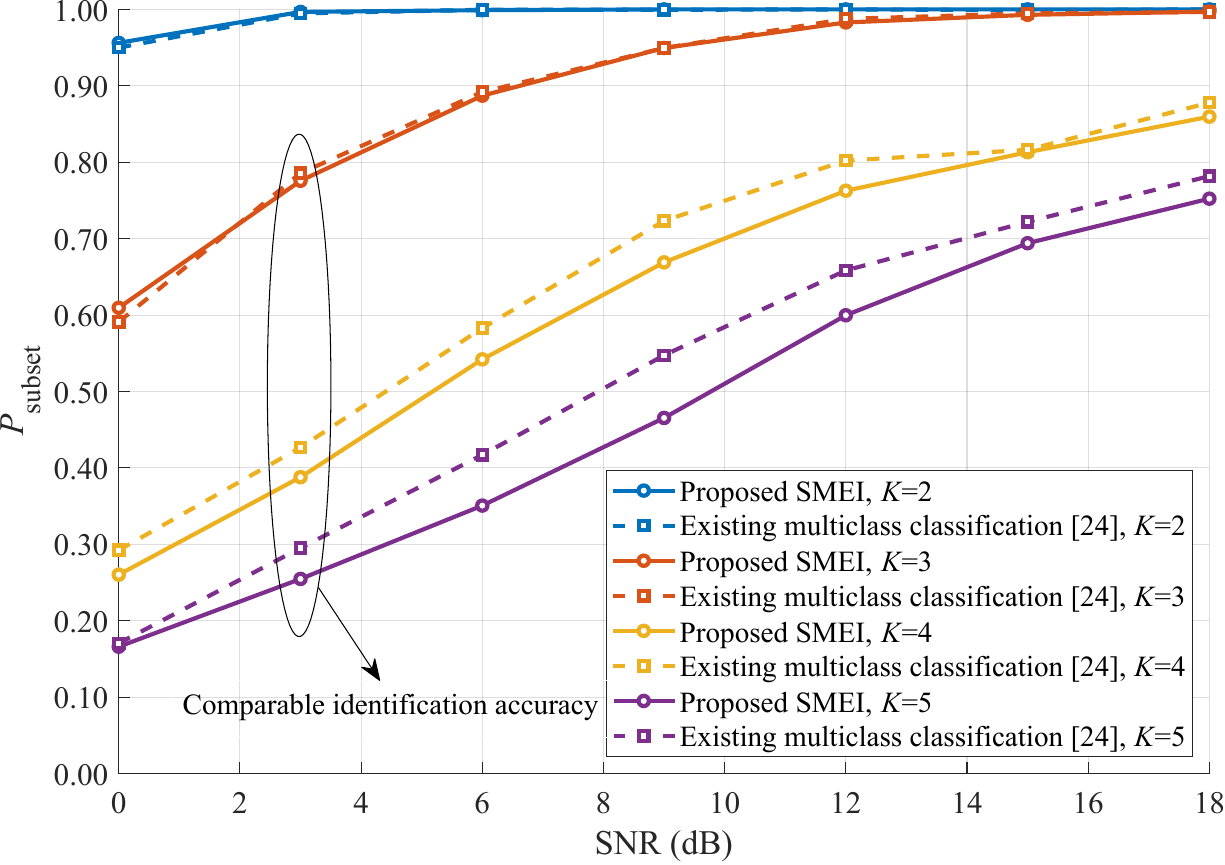}
	\caption{Comparison of SMEI and existing multiclass classification \cite{Sankhe2019ORACLE} on $P_{\text{subset}}$ with respect to the maximum number of emitters \(K\). Both approaches achieve comparable identification accuracy across different SNRs.}
	\label{fig:pc_subset_allk}
\end{figure}

\begin{figure}[]
	\centering
	\includegraphics[width=0.244\textwidth]{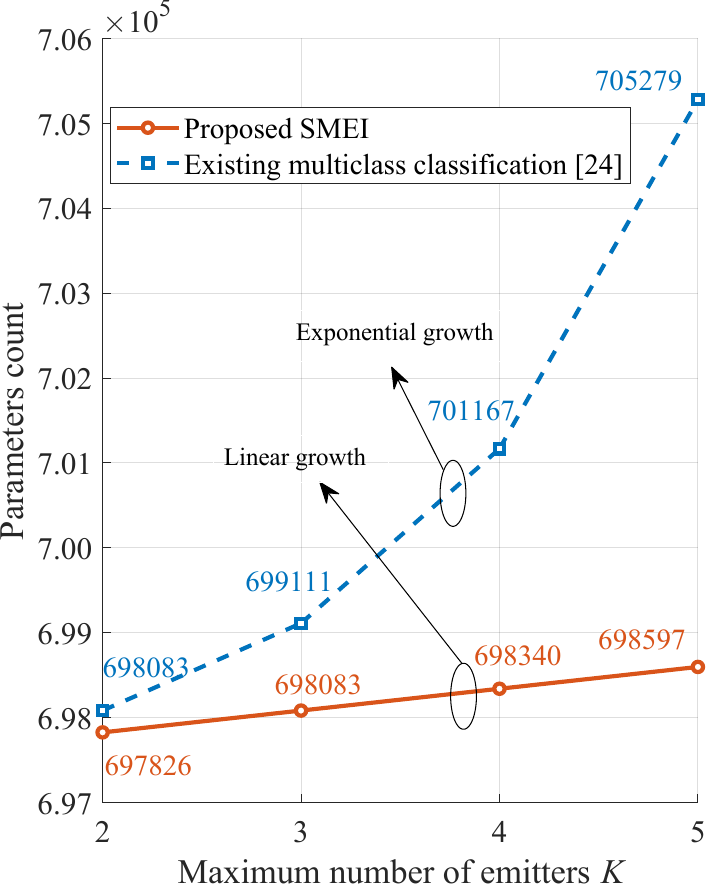}
	\hspace{-0.012\textwidth}
	\includegraphics[width=0.244\textwidth]{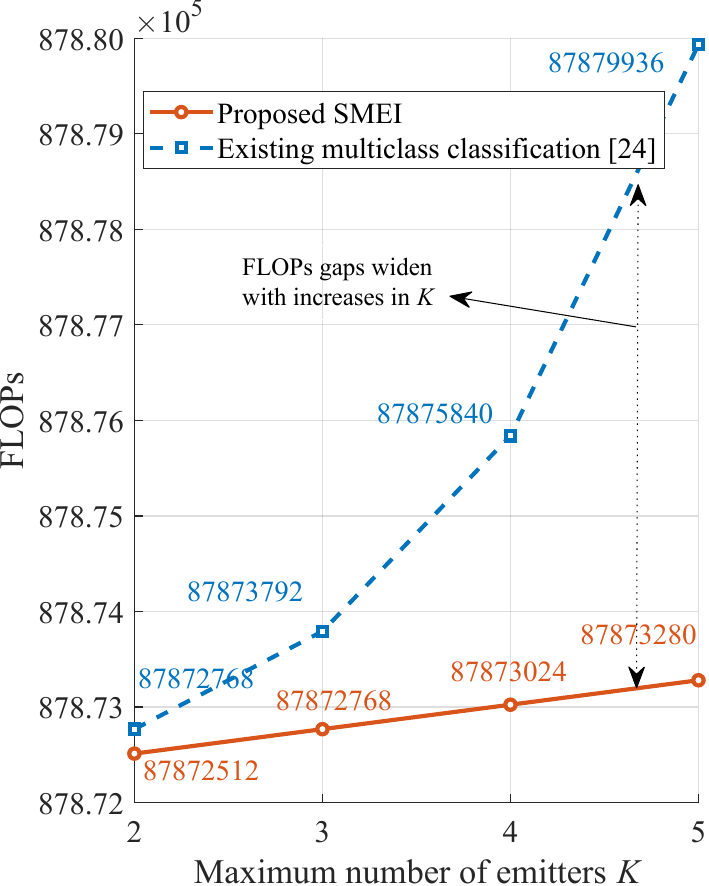}
	\caption{Comparison of proposed SMEI and existing multiclass classification \cite{Sankhe2019ORACLE} regarding model complexity with respect to the maximum number of emitters \(K\). (a) Parameter count; (b) FLOPs. The parameter count and FLOPs of SMEI increase linearly with \(K\), whereas the existing method exhibits exponential growth.}
	\label{fig:params_allk}
\end{figure}
\setcounter{table}{3}
\begin{table}[t]
	\centering
	\caption{\textsc{Comparison of Asymptotic Complexity between proposed SMEI and Existing Multiclass Classification}}
	\label{tab:complexity-comparison}
	\begin{tabular}{lm{2.28cm}m{3cm}}
		\toprule
		Complexity metric & Proposed SMEI & Multiclass methods \cite{Sankhe2019ORACLE} \\
		\midrule
		Parameter count & $\mathcal{O}\left(C_0^2 + C_0K\right)$ & $\mathcal{O}\left(C_0^2 + C_02^K\right)$ \\
		FLOPs & $\mathcal{O}\left(C_0^2l_{\mathrm{in}} + C_0K\right)$ & $\mathcal{O}\left(C_0^2l_{\mathrm{in}} + C_02^K\right)$ \\
		\bottomrule
	\end{tabular}
\end{table}

Fig.~\ref{fig:pc_subset_allk} compares the identification accuracy of the proposed SMEI with existing multiclass methods. It is observed that SMEI achieves comparable identification accuracy across varying SNR and maximum emitter numbers $K$. Furthermore, let the input signal length be \( l_{\mathrm{in}} \) and the base channel number be \( C_0 \). Table~\ref{tab:complexity-comparison} compares the model complexity in terms of parameter count and Floating Point Operations (FLOPs). The results show that, compared to existing multiclass methods that treat multi-emitter identification as a single multiclass classification problem with $2^K-1$ classes, the proposed SMEI framework significantly reduces both metrics. Specifically, the parameter complexity decreases from $\mathcal{O}(C_0^2 + C_02^K)$ to $\mathcal{O}(C_0^2 + C_0K)$, while the FLOPs are reduced from $\mathcal{O}(C_0^2l_{\mathrm{in}} + C_02^K)$ to $\mathcal{O}(C_0^2l_{\mathrm{in}} + C_0K)$. This improvement results from reformulating the original multiclass structure into a multi-label framework with $K$ independent binary classification heads, thus avoiding exponential growth of the output layer. 

Additionally, the theoretical complexity advantage is fully validated by experimental results in Fig.~\ref{fig:params_allk}. It can be seen that the proposed SMEI scheme exhibits linear growth in both parameter count and FLOPs as $K$ increases. For instance, when $K$ rises from 2 to 5, the parameter count increases by only 771, while FLOPs grow by merely 0.41M. In contrast, the existing multiclass method demonstrates exponential growth under the same conditions, with the parameter count reaching 705,279 and FLOPs increasing by 7.07M. These results indicate that the proposed SMEI scheme not only achieves high identification accuracy comparable to existing methods but also maintains computational complexity that scales linearly with the maximum number of emitters, significantly enhancing scalability in distributed networks.
\begin{figure}[t]
	\centering
	% 子图 (a)，宽度设置为 0.25\textwidth
	\subfloat[Overlap = 0\%]{%
		\includegraphics[width=0.2405\textwidth]{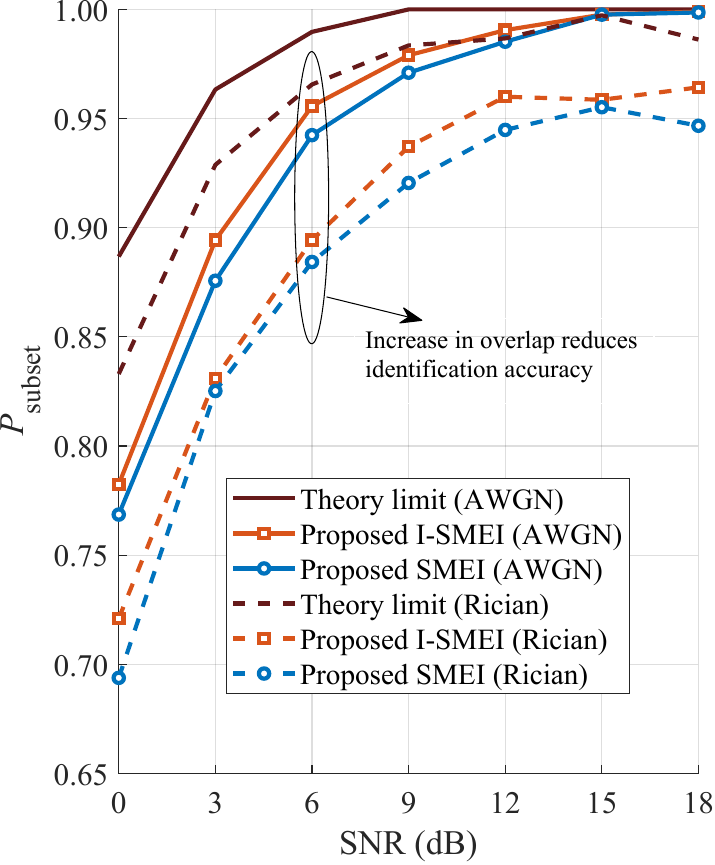}%
		\label{fig:subset_acc_100}
	}
	\hspace{-0.013\textwidth} % 调整间距
	% 子图 (b)，宽度设置为 0.25\textwidth
	\subfloat[Overlap = 50\%]{%
		\includegraphics[width=0.2405\textwidth]{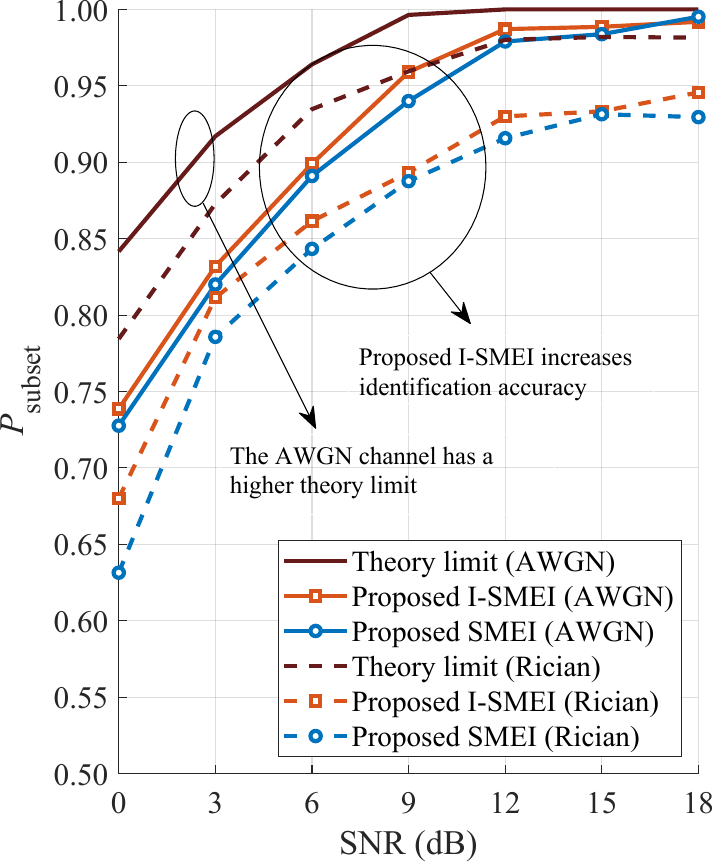}%
		\label{fig:subset_acc_50}
	}
	\caption{The variation of $P_{\text{subset}}$ w.r.t SNR at different overlap ratios (\(K=3\)). This trend shows that I-SMEI improves identification accuracy over SMEI and has greater potential to approach the theoretical limit in AWGN channels.}
	\label{fig:subset_acc_all}
\end{figure}
\begin{figure}[]
	\centering
	
	% 子图 (a)
	\subfloat[Overlap = 0\%]{%
		\includegraphics[width=0.2405\textwidth]{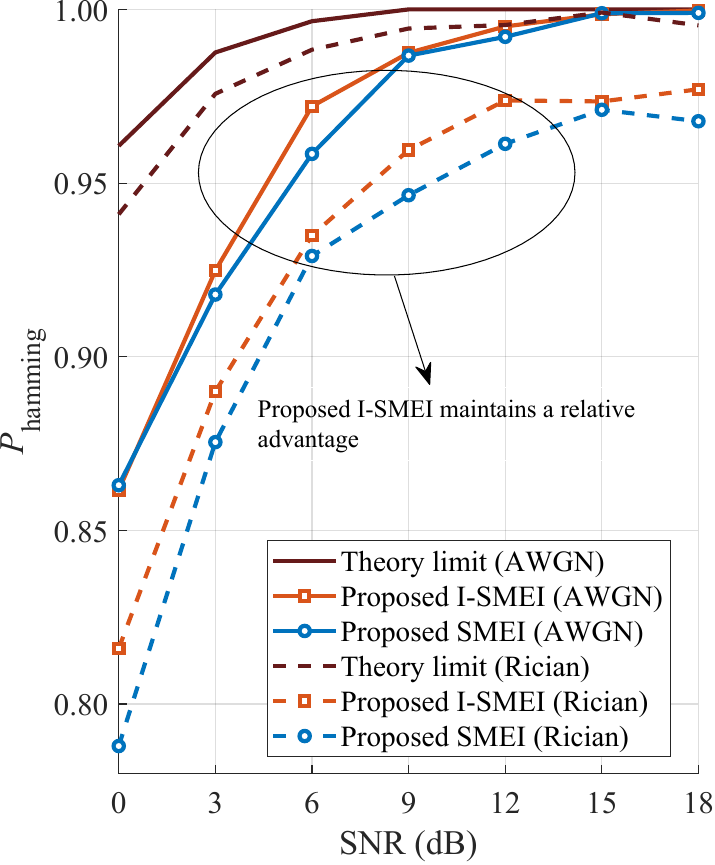}%
		\label{fig:hamming_acc_100}
	}
	\hspace{-0.013\textwidth} % 调整子图之间的间距
	% 子图 (b)
	\subfloat[Overlap = 50\%]{%
		\includegraphics[width=0.2405\textwidth]{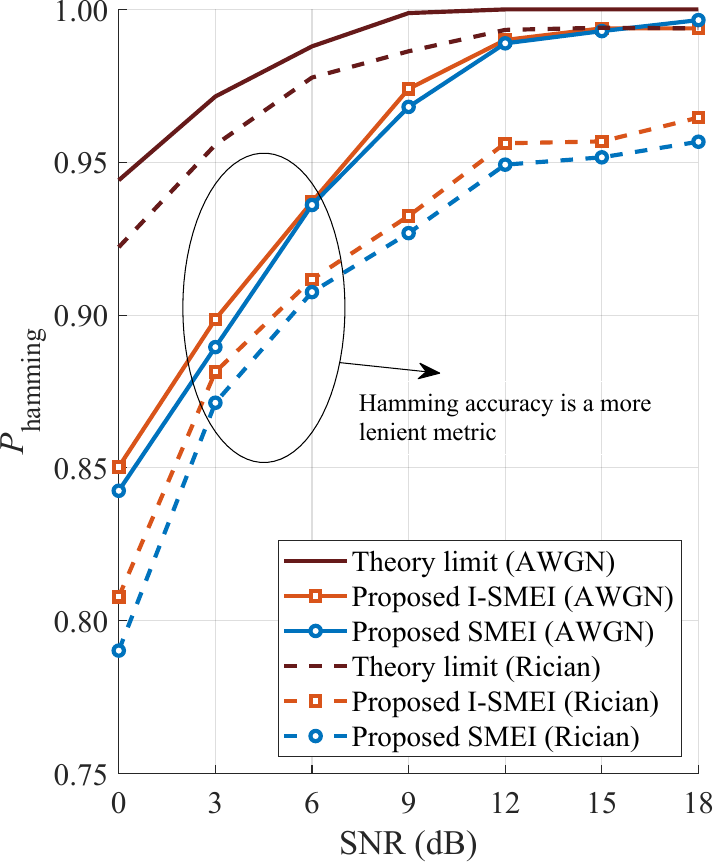}%
		\label{fig:hamming_acc_50}
	}
	\caption{The variation of $P_{\text{hamming}}$ w.r.t. SNR at different overlap ratios (\(K=3\)), indicating that I-SMEI maintains a performance advantage even under this more relaxed metric compared to $P_{\text{subset}}$.}
	\label{fig:hamming_acc_all}
\end{figure}

\begin{observation}
The proposed I-SMEI demonstrates improved performance over SMEI with both evaluation metrics nearly reach their theoretical limits at high SNR. (cf. Figs.~\ref{fig:subset_acc_all} and \ref{fig:hamming_acc_all})
\end{observation}
Fig.~\ref{fig:subset_acc_all} and Fig.~\ref{fig:hamming_acc_all} compare the performance of I-SMEI and SMEI in terms of subset accuracy and Hamming accuracy across different signal overlap ratios and channel conditions. It can be seen from the comparison between Fig.~\ref{fig:subset_acc_all}(a) and Fig.~\ref{fig:subset_acc_all}(b) that the identification accuracy of both methods decreases significantly as the overlap ratio increases from 0\% to 50\%. However, within each subfigure, I-SMEI consistently outperforms SMEI, particularly at high SNR, where it approaches the theoretical limit.

The network architectures of the proposed I-SMEI and proposed SMEI share the same core backbone, differing only in the cross-sample interaction module. Here, this performance improvement can be attributed to the multi-head attention mechanism introduced in I-SMEI, which enables cross-sample message passing and effectively leverages correlations inherent in distributed emitter environments. This allows each emitter to adaptively integrate feature information from neighboring sources, thereby enabling mutual feature-level enhancement. Notably, as shown in Fig.~\ref{fig:hamming_acc_all}, although Hamming accuracy is a more lenient evaluation metric, its trend is consistent with that of subset accuracy, and I-SMEI maintains a performance advantage under this metric as well. These results indicate that across various overlap ratios and SNRs, the proposed I-SMEI steadily improves identification accuracy compared with the proposed SMEI.

\begin{figure}[t]
	\centering
	
	% 图 (a)
	\includegraphics[width=0.245\textwidth]{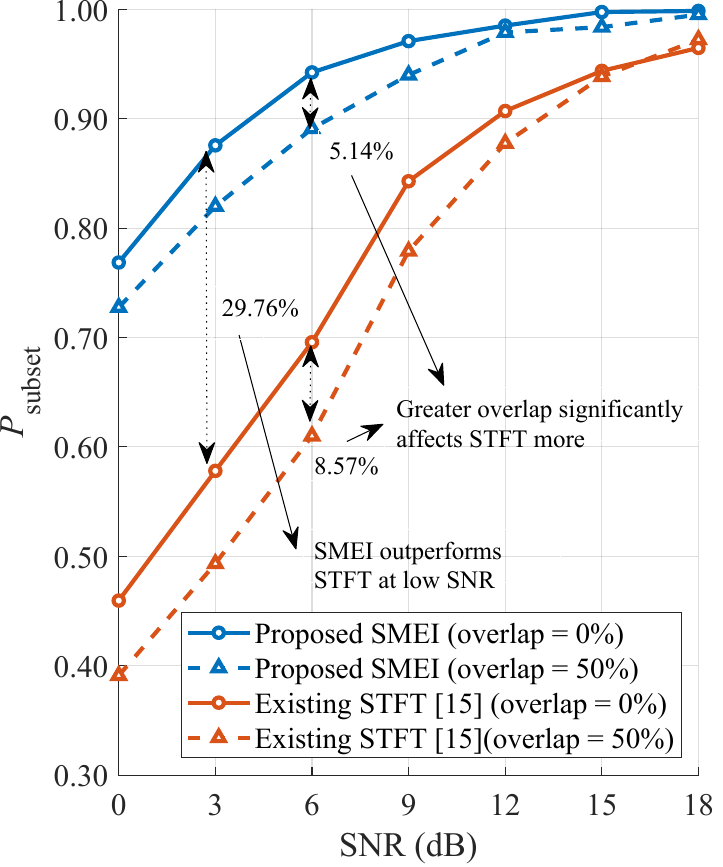}
	\hspace{-0.015\textwidth} % 调整图片之间的间距
	% 图 (b)
	\includegraphics[width=0.245\textwidth]{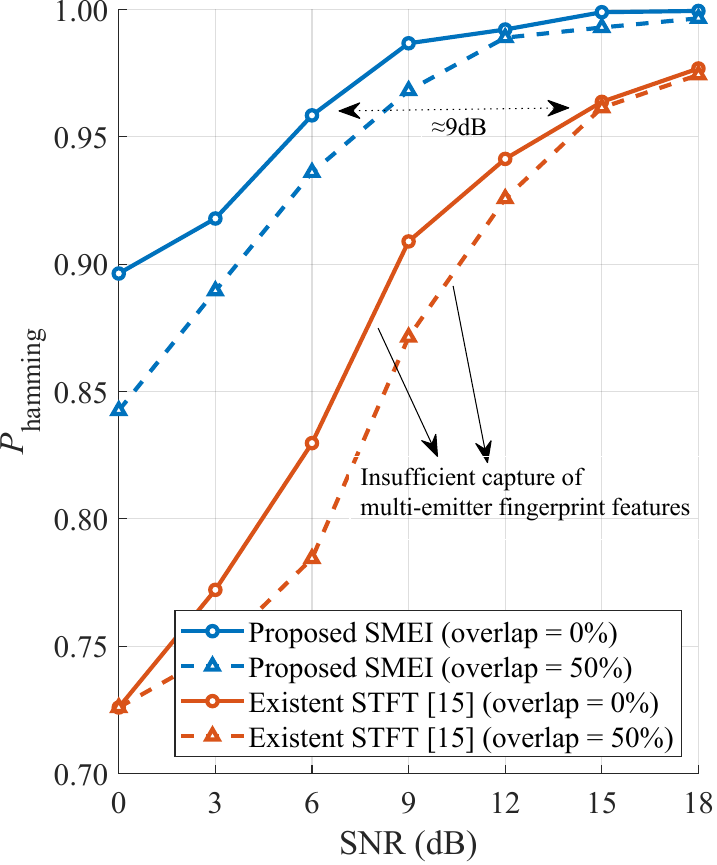}
	
	\caption{The comparison between the proposed SMEI and the existing STFT \cite{Yu2024DroneRFa} scheme (\(K=3\)). The proposed SMEI scheme demonstrates a significant advantage in identification accuracy, particularly at low SNRs, while the existing STFT scheme exhibits limitations in representing multi-emitter fingerprint features, leading to a more pronounced performance degradation at higher overlap ratios than proposed SMEI.}
	\label{fig:Hanming_subset_acc_AWGN_0_AND_50}
\end{figure}

\begin{figure}[t]
	\centering
	
	% 图 (a)
	\includegraphics[width=0.243\textwidth]{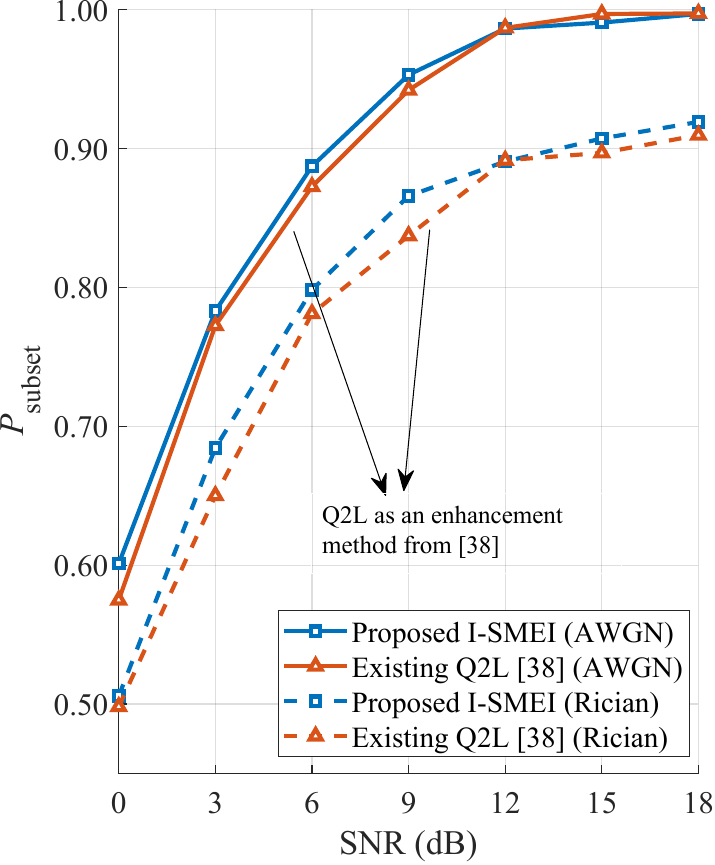}
	\hspace{-0.01\textwidth} % 调整图片之间的间距
	% 图 (b)
	\includegraphics[width=0.243\textwidth]{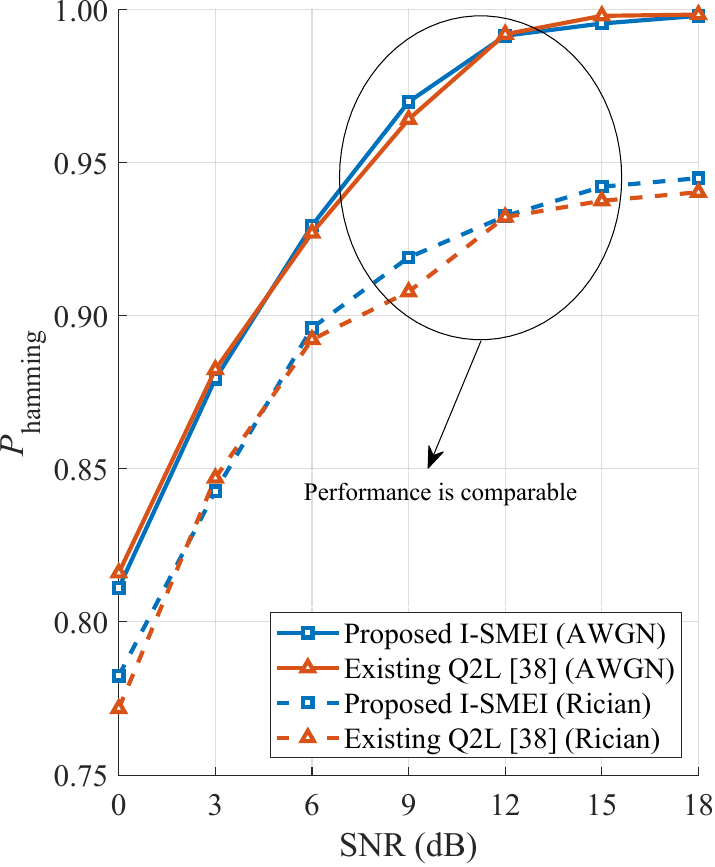}
	
	\caption{The comparison between the proposed I-SMEI and the existing Q2L \cite{Liu2021Query2Label} scheme (\(K=3\), overlap = 100\%) reveals that ours achieves superior robustness and identification accuracy across multiple environments, while the existing Q2L delivers comparable performance only on \(P_{\text{hamming}}\).}
	\label{fig:acc_comparison_ratio100}
\end{figure}

\begin{figure}[t]
	\centering
%	\captionsetup[subfloat]{font=tiny,labelfont=sf,textfont=sf}
	
	% 子图 (a)
	\subfloat[\(P_{\text{subset}}\) at SNR = 12 dB]{%
		\includegraphics[width=0.42\textwidth]{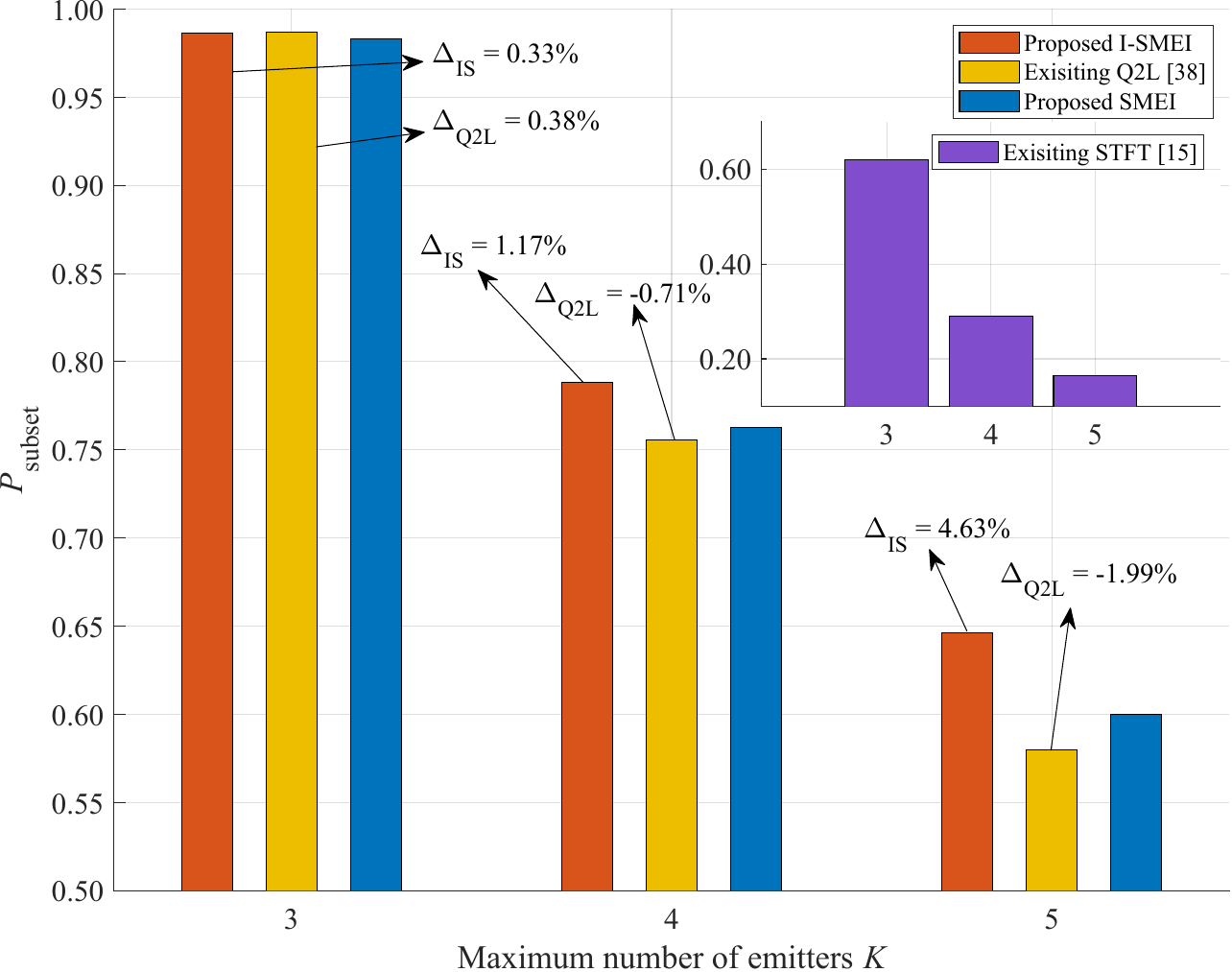}%
		\label{fig:subset_acc}
	}
	\hfill
	% 子图 (b)
	\subfloat[\(P_{\text{hamming}}\) at SNR = 12 dB]{%
		\includegraphics[width=0.42\textwidth]{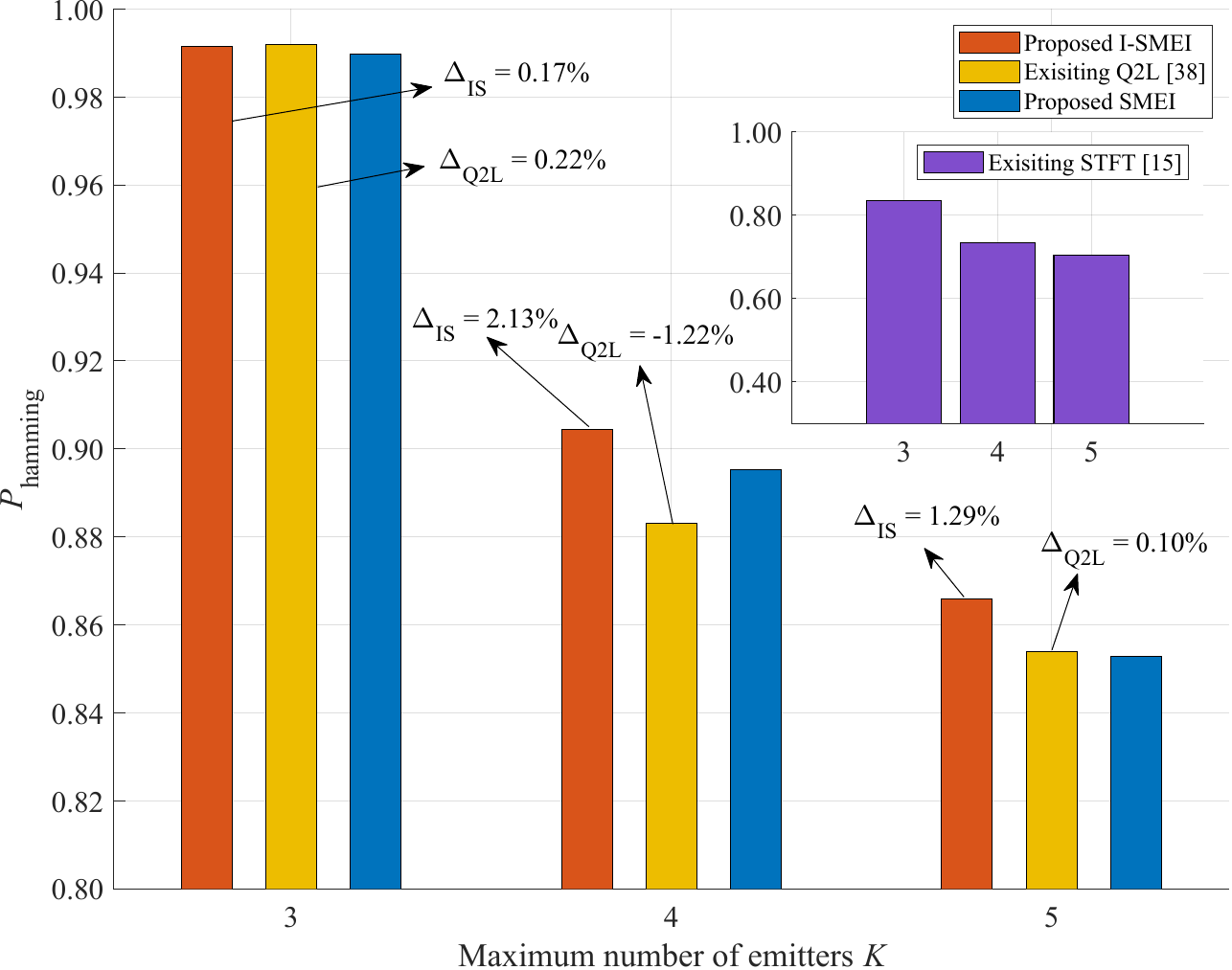}%
		\label{fig:hamming_acc}
	}
	
	\caption{The \(P_{\text{subset}}\) and \(P_{\text{hamming}}\) vs. maximum number of emitters \(K\) for different methods. The proposed I-SMEI shows significant gains in both metrics compared to the proposed SMEI across all \(K\) values, while the existing Q2L exhibits poor robustness and inconsistent gains, and the existing STFT method performs the worst across all \(K\) values.}
	\label{fig:accuracy_comparison}
\end{figure}

\begin{observation}
For the fundamental scheme, the proposed SMEI scheme demonstrates significant advantages in identificatioin accuracy over the existing STFT scheme \cite{Yu2024DroneRFa}. Meanwhile, for the enhanced scheme, the I-SMEI scheme also exhibits superior performance compared to the current Q2L scheme.  (cf. Figs.~\ref{fig:Hanming_subset_acc_AWGN_0_AND_50} and \ref{fig:acc_comparison_ratio100})
\end{observation}

Fig.~\ref{fig:Hanming_subset_acc_AWGN_0_AND_50} compares the performance of the proposed SMEI scheme with the conventional STFT-based method in multi-emitter overlapping signal identification. The STFT method shows significant sensitivity to signal overlap ratios because it struggles to represent overlapping features. Specifically, when the overlap ratio increases from 0\% to 50\%, the STFT method shows a notable performance degradation. In contrast, the SMEI scheme demonstrates more stable performance characteristics, with its subset accuracy fluctuation being 3.43\% lower than that of the STFT method at an SNR of 6 dB. Furthermore, Fig.~\ref{fig:acc_comparison_ratio100} presents the comparison results between the I-SMEI scheme and the Q2L method in terms of identification accuracy. It can be observed that although both enhancement schemes achieve accuracies, I-SMEI maintains a more stable advantage in overall performance. These results indicate that the proposed SMEI and I-SMEI schemes offerofferpossess systematic advantages over baseline methods for overlapping signal identification in distributed environments.

\begin{observation}
The proposed I-SMEI approach outperforms both the existing Q2L and the proposed SMEI methods under the overlap = 100\% case, whereas the existing STFT method performs much worse. (cf. Fig.~\ref{fig:accuracy_comparison})
\end{observation}
Fig.~\ref{fig:accuracy_comparison} presents the identification accuracy of different methods under complete overlap conditions in AWGN channels at SNR = 12 dB. It can be seen from Fig.~\ref{fig:accuracy_comparison} (a) that the STFT method performs the worst, with its accuracy decreasing significantly as the number of emitters increases. Here, a detailed numerical comparison reveals that I-SMEI consistently enhances the subset accuracy over SMEI, with relative improvements of 0.33\%, 1.17\%, and 4.63\% at \(K\) = 3, 4, and 5, respectively. In contrast, Q2L exhibits unstable performance relative to SMEI, achieving a 0.38\% gain at \(K\) = 3 but declining to -0.71\% and -1.99\% at \(K\) = 4 and 5, respectively. Consequently, while Q2L shows an advantage over SMEI at \(K\) = 3, this advantage gradually diminishes and eventually reverses as \(K\) increases. 

Fig.~\ref{fig:accuracy_comparison} (b) reveals a similar trend in Hamming accuracy. In comparison, I-SMEI outperforms SMEI on both metrics, and its advantage over SMEI persists even as the number of emitters increases. The comparison demonstrates that I-SMEI achieves better identification accuracy than both Q2L and SMEI, confirming its effectiveness in distributed networks with growing numbers of devices.
\section{Conclusion}
\label{sec7}
In this paper, we investigate the SMEI problem in distributed wireless networks and propose a low-complexity solution via multi-label learning. Conventional methods designed for single-emitter scenarios struggle with overlapping signals from concurrent transmissions by distributed emitters. To address this challenge, we reformulate SMEI as a multi-label learning problem and derive theoretical performance bounds 
via Fano's inequality.  Specifically, the main contributions are threefold. First, we establish theoretical performance bounds for SMEI via Fano's inequality, providing a rigorous benchmark for method evaluation. Second, we reformulate SMEI 
as a multi-label learning problem for the first time, treating each emitter as an independent label. This reduces the output dimensionality from $\mathcal{O}(2^K)$ to $\mathcal{O}(K)$, substantially reducing the number of model parameters and FLOPs. Third, we propose I-SMEI, which incorporates a message-passing module via multi-head attention to effectively capture emitter features in correlated signal combinations. Experimental results show that SMEI maintains comparable accuracy to multiclass classification while significantly reducing complexity. Furthermore, both I-SMEI and SMEI offer substantial advantages over existing STFT-based and Q2L schemes in terms of identification accuracy and robustness.

For future research, we will focus on exploring more expressive deep network architectures for distributed networks, with the goal of closing the gap with theoretical limits under low SNR and other practical impairments. Concurrently, to meet real-time identification demands in high-density distributed emission scenarios, we will develop more efficient algorithmic frameworks to optimize computational efficiency and resource utilization for devices with limited capabilities.

{\appendices \section{Proof of \textit{Lemma}~\ref{lemma01}} \label{appendix:fano_derivation}
This lemma derives theoretical upper bounds on the subset and Hamming accuracies via Fano's inequality. To begin the derivation, the label matrix is defined by
\begin{align}
	\bm{\varLambda} &= [\bm{\lambda}^{(1)}, \bm{\lambda}^{(2)}, \ldots, \bm{\lambda}^{(N)}],
\end{align}
where $\bm{\varLambda}$ contains $|\bm{\varLambda}| = 2^{K} - 1$ distinct label combinations, excluding the all-zero case, and the general form of Fano's inequality~\cite{Wang2016UserCapacity} is given as
\begin{equation}
	H\left(\boldsymbol{\lambda}\mid\boldsymbol{y}\right) \leq H\left(P_e\right) + P_e \log\left(|\bm{\varLambda}| - 1\right),
	\label{eq:Fano_eq_app}
\end{equation}
where $P_e$ denotes the classification error probability. Then, the subset accuracy is directly related to the error probability by
\begin{align}
	P_{\text{subset}} &= 1 - P_e, \label{eq:1_pe_pc_app} \\
	H\left(P_e\right) &= H\left(P_{\text{subset}}\right), \label{eq:he_pe}
\end{align}
where the binary entropy function is expressed as
\begin{equation}
	\begin{aligned}
		H\left(P_{\text{subset}}\right) = &-P_{\text{subset}} \log\left(P_{\text{subset}}\right) \\
		&- \left(1 - P_{\text{subset}}\right) \log\left(1 - P_{\text{subset}}\right).
	\end{aligned}
	\label{eq:hc_subset}
\end{equation}

By substituting~\eqref{eq:1_pe_pc_app} and~\eqref{eq:hc_subset} into~\eqref{eq:Fano_eq_app}, the inequality the inequality is given by
\begin{equation}
	\begin{aligned}
		& \log\left(|\bm{\varLambda}|-1\right) P_{\text{subset}} + P_{\text{subset}} \log\left(P_{\text{subset}}\right) \\
		& \quad + \left(1-P_{\text{subset}}\right) \log\left(1-P_{\text{subset}}\right) \\
		& \leq \log\left(|\bm{\varLambda}|-1\right) - H\left(\boldsymbol{\lambda} \mid \boldsymbol{y}\right),
		\label{eq:P_C_D_app}
	\end{aligned}
\end{equation}

To simplify the left-hand side of~\eqref{eq:P_C_D_app}, we define an auxiliary function as
\begin{equation}
	g(p) = \log\left(|\bm{\varLambda}|-1\right)p + p\log(p) + (1-p)\log(1-p),
	\label{eq:g_p_function}
\end{equation}
where $p \in (0,1)$. With this notation, the left-hand side of~\eqref{eq:P_C_D_app} can be expressed as $g(P_{\text{subset}})$. Subsequently, the derivative of $g(p)$ with respect to $p$ is written as
\begin{equation}
	g'(p) = \log\!\left(\frac{(|\bm{\varLambda}|-1)p}{1-p}\right),
	\label{eq:g_derivative}
\end{equation}
where the stationary condition is obtained by setting $g'(p)=0$, so we can obtain
\begin{equation}
	\frac{(|\bm{\varLambda}|-1)p}{1-p} = 1,
	\label{eq:stationary_condition}
\end{equation}
which yields
\begin{equation}
	p = \frac{1}{|\bm{\varLambda}|}.
	\label{eq:min_point}
\end{equation}

To verify that this stationary point is indeed a minimum, we examine the second derivative, which is given by
\begin{equation}
	g''(p) = \frac{1}{p} + \frac{1}{1-p} > 0,
	\label{eq:second_derivative}
\end{equation}
where the positivity holds for all $p \in (0,1)$, showing that $g(p)$ is strictly convex, and thus $p = 1/|\bm{\varLambda}|$ corresponds to its unique global minimum.

To proceed with deriving the upper bound on $P_{\text{subset}}$, we need to express the conditional entropy $H(\boldsymbol{\lambda}\mid\boldsymbol{y})$ in terms of quantities that can be bounded. According to the definition of mutual information, we have
\begin{equation}
	I\left(\boldsymbol{\lambda};\boldsymbol{y}\right) = H\left(\boldsymbol{\lambda}\right) - H\left(\boldsymbol{\lambda}\mid\boldsymbol{y}\right),
	\label{eq:I_XY_app}
\end{equation}
and substituting~\eqref{eq:I_XY_app} into~\eqref{eq:P_C_D_app} leads to
\begin{equation}
	g(P_{\text{subset}}) \leq \log\!\left(|\bm{\varLambda}| - 1\right) - H(\boldsymbol{\lambda}) + I(\boldsymbol{\lambda};\boldsymbol{y}),
	\label{eq:I_HXY_app}
\end{equation}
where the constant term $C$ is defined as
\begin{equation}
	C = \log(|\bm{\varLambda}| - 1) - H(\boldsymbol{\lambda}) + I(\boldsymbol{\lambda};\boldsymbol{y}).
	\label{eq:C_definition}
\end{equation}

Since $g(p)$ is convex and \( g(P_{\text{subset}}) \leq C \), the upper bound of $P_{\text{subset}}$ can be written as
\begin{equation}
	\widetilde{P}_{\text{subset}} = 
	\begin{cases}
		1/|\bm{\varLambda}|, & C \leq g(1/|\bm{\varLambda}|), \\
		\min(1, g^{-1}(C)), & C > g(1/|\bm{\varLambda}|).
	\end{cases}
	\label{eq:subset_upper_bound}
\end{equation}

To derive the upper bound of Hamming accuracy from the subset accuracy bound, we assume that the label vectors follow a uniform distribution, i.e.
\begin{equation}
	\mathbb{P}(\boldsymbol{\lambda} = \boldsymbol{\lambda}^{(i)}) = \frac{1}{|\bm{\varLambda}|}, 
	\label{eq:uniform_distribution}
\end{equation}
where $i \in \{1, 2, \ldots, |\bm{\varLambda}|\}$, and consequently, the entropy of $\boldsymbol{\lambda}$ is computed as
\begin{equation}
	H(\boldsymbol{\lambda}) = \log(|\bm{\varLambda}|).
	\label{eq:entropy_lambda}
\end{equation}

Furthermore, we assume that each label has the same prediction accuracy, i.e.
\begin{equation}
	\mathbb{P}(\hat{\lambda}_{k_1} = \lambda_{k_1}) = \mathbb{P}(\hat{\lambda}_{k_2} = \lambda_{k_2}),
	\label{eq:probability_equation_app}
\end{equation}
where $k_1, k_2 \in \{1, 2, \ldots, K\}$. Under this assumption, the probability of all $K$ labels being correct is given by
\begin{equation}
	P_{\text{subset}} = (P_{\text{hamming}})^K.
	\label{eq:subset_hamming_relation}
\end{equation}
where applying~\eqref{eq:subset_upper_bound} yields the theoretical upper bound for the Hamming accuracy, i.e.
\begin{equation}
	\widetilde{P}_{\text{hamming}} = (\widetilde{P}_{\text{subset}})^{1/K}.
	\label{eq:hamming_upper_bound}
\end{equation}

This completes proof of \textit{Lemma}~\ref{lemma01}.\hfill $\square$}

\vfill

\end{document}